\documentclass{article}

\usepackage[english]{babel}
\usepackage{pifont}

\usepackage[letterpaper,top=2cm,bottom=2cm,left=3cm,right=3cm,marginparwidth=1.75cm]{geometry}

\usepackage{amsmath}
\usepackage{graphicx}
\usepackage[colorlinks=true, allcolors=blue]{hyperref}

\usepackage{amsthm,mdframed}
\usepackage{amsfonts}
\usepackage{amssymb}
\usepackage{stmaryrd}
\usepackage{dsfont}
\usepackage{xcolor}

\usepackage{tcolorbox}
\usepackage{xspace}

\usepackage{float}

\usepackage{authblk}

\newtheorem{theorem}{Theorem}
\newtheorem{proposition}{Proposition}
\newtheorem{problem}{Problem}
\newtheorem{lemma}{Lemma}
\newtheorem{corollary}{Corollary}
\newtheorem*{remark*}{Remark}
\newtheorem{definition}{Definition}

\newcommand{\andre}[1]{\textcolor{blue}{[A : #1]}}

\newcommand{\Iint}[2]{\llbracket #1 , #2 \rrbracket}
\newcommand{\Tr}{\operatorname{Tr}}

\newcommand{\ket}[1]{|#1\rangle}
\newcommand{\bra}[1]{\langle#1|}
\newcommand{\ketbra}[2]{|#1\rangle\langle#2|}
\newcommand{\braket}[2]{\langle #1 | #2 \rangle}
\newcommand{\altketbra}[1]{\ketbra{#1}{#1}}
\newcommand{\kb}[1]{\altketbra{#1}}
\newcommand{\norm}[1]{\ensuremath{\lvert  \kern-1pt\lvert #1 \rvert \kern-1pt \rvert}}

\newcommand{\F}{\mathbb{F}}

\newcommand{\Hm}{\mathbf{H}}
\newcommand{\Gm}{\mathbf{G}}
\newcommand{\Mm}{\mathbf{M}}
\newcommand{\Pm}{\mathbf{P}}

\newcommand{\xv}{\mathbf{x}}
\newcommand{\yv}{\mathbf{y}}
\newcommand{\cv}{\mathbf{c}}
\newcommand{\sv}{\mathbf{s}}
\newcommand{\av}{\mathbf{a}}
\newcommand{\bv}{\mathbf{b}}
\newcommand{\vv}{\mathbf{v}}
\newcommand{\wv}{\mathbf{w}}
\newcommand{\ev}{\mathbf{e}}
\newcommand{\uv}{\mathbf{u}}
\newcommand{\tv}{\mathbf{t}}

\newcommand{\MAX}{\mathrm{MAX}}

\usepackage{thmtools}
\usepackage{thm-restate}

\newcommand{\G}{\mathcal{G}}
\renewcommand{\P}{\mathcal{P}}
\newcommand{\I}{\mathcal{I}}
\newcommand{\COMMENT}[1]{}

\newcommand{\trp}[1]{{#1}^{\intercal}}
\newcommand{\trsp}[1]{{#1}}
\newcommand{\eqdef}{\triangleq}

\newcommand{\ie}{\textit{i.e.}\xspace}

\newcommand{\E}{\mathbb{E}}
\newcommand{\zerov}{\mathbf{0}}
\newcommand{\bF}{\overline{F}}

\newcommand{\hy}{\widehat{\yv}}
\newcommand{\hyz}{\widehat{\yv_0}}

\renewcommand{\aa}{\mathcal{A}}

\newcommand{\hFH}{\widehat{F_\Hm}}
\newcommand{\iv}{\mathbf{i}}
\newcommand{\jv}{\mathbf{j}}
\newcommand{\lv}{\mathbf{l}}

\newcommand{\rv}{\mathbf{r}}

\newcommand{\betav}{\mathbf{\beta}}

\newcommand{\eps}{\varepsilon}

\newcommand{\FHy}{\{F_{\Hm,\yv}\}}

\newcommand{\Z}{\mathbb{Z}}

\usepackage{bbm}
\newcommand{\one}{\mathbbm{1}}

\newcommand{\FF}{\{F_{\Hm,\yv}\}}
\renewcommand{\SS}{\{\ket{\psi_{\xv}}\}}

\newcommand{\D}{\mathcal{D}}

\newcommand{\C}{\mathcal{C}}

\newcommand{\triple}[3]{\bra{#1}#2\ket{#3}}
\newcommand{\hi}{\widehat{\iv}}
\newcommand{\hj}{\widehat{\jv}}
\newcommand{\hx}{\widehat{\xv}}

\newcommand{\htv}{\widehat{\tv}}

\newcommand{\hwl}{\widehat{\wv_\lv}}

\newcommand{\hv}{\mathbf{h}}
\newcommand{\halpha}{\widehat{\alpha}}

\newcommand{\SLPN}{\ensuremath{\text{S-}\ket{\text{LPN}}}}
\newcommand{\SLWE}{\ensuremath{\text{S-}\ket{\text{LWE}}}}

\newcommand{\onev}{\mathbf{1}}

\newtcolorbox{problemBox}[1][]{%
	sharp corners,         
	colframe=black,        
	colback=gray!5,        
	coltitle=black,        
	colbacktitle=white,    
	fonttitle=\bfseries,   
	title={\centering Quantum Decoding Problem a.k.a. $\SLPN$}, 
	boxsep=2pt,        
	left=4pt,          
	right=4pt,         
	top=4pt,           
	bottom=4pt,        
	#1                     
}

\usepackage{yhmath}

\title{Fine-Grained Unambiguous Measurements}
\author[$\dagger$]{Quentin Buzet}
\author[$\dagger$]{André Chailloux}
\affil[$\dagger$]{Inria de Paris}
\date{}

\begin{document}
\maketitle

\begin{abstract}

Unambiguous measurements play an important role in quantum information, with applications ranging from quantum key distribution to quantum state reconstruction. Recently, such measurements have also been used in quantum algorithms based on Regev’s reduction. The key problem for these algorithms is the $\SLWE$ problem in the lattice setting and the Quantum Decoding Problem in the code setting. A key idea for addressing this problem is to use unambiguous measurements to recover $k$ coordinates of a code (or lattice) element $\xv$ from a quantum state $\ket{\psi_{\xv}}$, which corresponds to a noisy word $\xv$ with errors in quantum superposition. However, a general theoretical framework to analyze this approach has been lacking.

In this work, we introduce the notion of \emph{fine-grained unambiguous measurements}. Given a family of states $\{\ket{\psi_\xv}\}_{\xv \in \F_2^n}$, we ask whether there exist measurements that can return, with certainty, $k$ parities about $\xv$. We study this question in the setting of symmetric states, which naturally arises in the Quantum Decoding Problem. We show that determining the maximal number of parities that a measurement can output can be formulated as a linear program, and we use its dual formulation to derive several upper bounds. In particular, we establish necessary and sufficient conditions for the existence of fine-grained unambiguous measurements and prove impossibility results showing in particular that such measurements cannot improve upon the approach of~\cite{CT24}. Finally, we discuss the implications of these findings for the Quantum Decoding Problem.

\end{abstract}

\newpage
\tableofcontents
\newpage

\section{Introduction}
\subsection{Context}\label{Section:1.1}
Unambiguous measurements play an important role in various areas of quantum information, ranging from the study of quantum key distribution~\cite{DJL00,KM19} to quantum state reconstruction~\cite{KOYJ23}. Such measurements have been extensively studied since the seminal works by Ivanovic~\cite{Iva87}, Dieks~\cite{Die88}, and Peres~\cite{Per88}. Recently, and this motivates our work, variants of unambiguous measurements have been used in quantum algorithms based on Regev's reduction~\cite{Reg09}. These measurements are used to solve the $\SLWE$ problem, which can then be used to find small dual lattice points~\cite{CLZ22}. This approach based on Regev's reduction has also been adapted to codes~\cite{DRT23,CT24} or structured codes~\cite{JSW+24,CT25}. 

In its simplest form, an unambiguous measurement can be described as follows. Assume we are given one of two possible states $\ket{\psi_0},\ket{\psi_1}$ chosen uniformly at random, and we want to determine which state we have. If the states are not orthogonal, it is impossible to determine which state we have with certainty. An unambiguous measurement is a measurement that always outputs the correct state, but may sometimes output ``I don't know’’, characterized by the $\bot$ outcome. An unambiguous measurement for two pure states can therefore be formally defined as follows:
\begin{definition}
	An unambiguous measurement for the states $\{\ket{\psi_0},\ket{\psi_1}\}$ is a three-outcome POVM\footnote{POVM stands for Positive Operator-Valued Measure and describes a general quantum measurement. $F_0,F_1,F_\bot$ here can be any positive semidefinite matrices such that $F_0 + F_1 + F_\bot = I$.} $\{F_0,F_1,F_\bot\}$, where $\forall i \in \{0,1\}$, $\Tr(F_i \kb{\psi_{1-i}}) = 0$. The success probability of this measurement is defined as $p \eqdef \frac{1}{2} \left(\Tr(F_0 \kb{\psi_0}) + \Tr(F_1 \kb{\psi_1})\right)$.
\end{definition}
We know there exists an unambiguous measurement for two pure states $\ket{\psi_0},\ket{\psi_1}$ with success probability $1 - |\braket{\psi_0}{\psi_1}|$, and this is optimal~\cite{JS95}. This problem has also been generalized to cases where the probabilities of receiving $\ket{\psi_0}$ and $\ket{\psi_1}$ differ~\cite{PW91,SHB01}, and to mixed states~\cite{RLvE03}. For two states, unambiguous measurements are well understood.

Another natural generalization is to consider a larger number of states. We are given one state from the set $\{\ket{\psi_x}\}_{x \in \Iint{1}{N}}$ with $N \ge 2$, and the goal is to recover $x$ from $\ket{\psi_x}$. An unambiguous measurement will always output the correct $x$ or $\bot$. In this more general setting, much less is known. Chefles~\cite{Che98} showed that unambiguous state discrimination is possible if and only if the states are linearly independent. If the states $\{\ket{\psi_\xv}\}$ exhibit certain symmetries, then more can be said about the problem.
We present only the case of $\F_2^n$ here, but refer to~\cite{CB98,CT24} for more general cases.
\begin{definition}\label{Definition:symmetric} 
	A set of states $S = \{\ket{\psi_{\xv}}\}_{\xv \in \F_2^n}$ is called symmetric iff. we can write
	$$ \ket{\psi_\zerov} = \sum_{\iv \in \F_2^n} \alpha_\iv \ket{\iv} \quad \text{and} \quad \forall \xv \in \F_2^n, \ \ket{\psi_\xv} = \sum_{\iv \in \F_2^n} \alpha_\iv \ket{\iv + \xv}.$$
\end{definition} 
This means each $\ket{\psi_\xv}$ is a shifted version of $\ket{\psi_\zerov}$. We work in $\F_2^n$, so we also define the Fourier basis $\ket{\hi} = \frac{1}{\sqrt{2^n}} \sum_{\jv \in \F_2^n} (-1)^{\iv \cdot \jv} \ket{\jv}$, where $\cdot$ is the canonical inner product in $\F_2^n$. Chefles and Barnett showed the following:

\begin{proposition}[\cite{CB98}]
	Let $S = \{\ket{\psi_\xv}\}_{\xv \in \F_2^n}$ be a set of symmetric states. There exists an unambiguous measurement for $S$ that succeeds with probability $P = 2^n \min_{\xv \in \F_2^n} |\braket{\hx}{\psi_\zerov}|^2$, and this is optimal.
\end{proposition}

This means we also have a complete answer for the performance of unambiguous measurements in the case of symmetric states.

\paragraph{An Unexpected Connection: Quantum Algorithms Based on Regev's Reduction.}
An exciting new line of research involves the study of quantum algorithms based on Regev's reduction, where unambiguous measurements play an important role. We present here a brief description of this connection for the case of codes.

We are given a binary\footnote{We restrict ourselves to the binary setting in this work, but the approach also applies to larger alphabets.} linear code $\C$ of dimension $k$ and length $n$. We can write $\C = \{\xv : \Mm \xv =  0$\}, for some known parity matrix\footnote{The parity matrix of a code is usually denoted $\Hm$ but we will use this letter in a somewhat different context, so we unconventionally call the parity matrix $\Mm$ in this section.} $\Mm \in \F_2^{(n-k) \times n}$. The main problem that has to be solved in these quantum algorithms is the following

\begin{problem}[Quantum Decoding Problem a.k.a. $\SLPN$]
	Given $\ket{\psi_{\xv}} = \sum_{\iv \in \F_2^n} \alpha_{\iv} \ket{\xv + \iv}$, for a random $\xv \in \C$, recover ${\xv}$. Here, there is a known description of the $\alpha_\iv$.
\end{problem}
 This is exactly a quantum state discrimination problem on symmetric states with the extra promise that $\xv$ is a codeword. It can be rephrased as recovering a codeword $\xv$ from a noisy version  $\ket{\psi_\xv} \eqdef \sum_{\ev \in \F_2^n} f(\ev)\ket{\xv + \ev}$ (which corresponds to $f(\ev) = \alpha_{\ev}$) where the error is on quantum superposition. An adaptation of Regev's reduction for codes~\cite{DRT23} can be informally stated as follows:

\begin{proposition}[Informal,~\cite{DRT23}]\label{Proposition:2}
	If we have an efficient algorithm to recover $\xv$ from $\ket{\psi_\xv}$, then we can efficiently find non-zero dual codewords in $\C^\bot$ with weight concentrated around the typical weight of the distribution associated to the $|\halpha_\iv|^2$, where $\halpha_\iv = \frac{1}{\sqrt{2^n}} \sum_{\jv \in \F_2^n} (-1)^{\iv \cdot \jv} \alpha_\jv$.
\end{proposition}

Here,  the typical weight of words associated to $|\halpha_\iv|^2$ can be understood as follows: if we measure $\sum_{\iv \in \F_2^n} \halpha_\iv \ket{\iv}$ in the computational basis, what will be the typical Hamming weight of the measured $\iv$?
Finding small words in a code (or its dual) is the main hardness assumption used in code-based cryptography. The above proposition also has equivalents in the realm of lattice-based cryptography, but we work in $\Z_q^n$ instead of $\F_2^n$, see for example~\cite{CH25}.

A natural question is how to incorporate the promise $\xv \in  \C$ for this quantum state discrimination problem. In~\cite{CT24}, the authors used the fact that one can recover the whole string $\xv = x_1,\dots,x_n$ from only $k(1+o(1))$ randomly chosen coordinates $x_i$ using Gaussian elimination. Indeed, since $\xv \in \C$, the relation $\Mm \xv = \zerov$ already gives $(n-k)$ parities of $\xv$, so if we can recover unambiguously $k$ additional parities of $\xv$ which are linearly independent with the ones given by $\Mm$ then one can recover $\xv$ using Gaussian elimination. The linear independence condition is satisfied with high probability if we recover $k(1+o(1))$ random parities of $\xv$. This motivates the following question

\begin{problem}\label{Problem:AA}
	Given a set $\{\ket{\psi_{\xv}}\}_{\xv \in \F_2^n}$ of symmetric states, how many parities of $\xv$ can we recover unambiguously?
\end{problem}

In~\cite{CT24}, the authors considered states of the form 
$$ \ket{\psi_\xv} = \sum_{\iv \in \F_2^n} \alpha_{\iv} \ket{\xv + \iv} =  \bigotimes_{i = 1}^n \left(\sqrt{1-t}\ket{x_i} + \sqrt{t}\ket{1-x_i}\right). $$
Here, each $\ket{\psi_\xv}$ is a product state over the $n$ single qubit registers. They showed that if we perform an unambiguous measurement on each qubit, we know we can recover each $x_i$ unambiguously with probability $2t^\perp$ where $t^\perp \eqdef \frac{1}{2} - \sqrt{t(1-t)}$. On the other hand, one can show that the typical weight associated to the $|\halpha_\iv|^2$ is $t^{\perp}n$. If we can recover $k = 2nt^\perp$ such coordinates (or slightly more to have total linear independence), then we can recover $\xv$. Plugging this into Proposition~\ref{Proposition:2} allows us to find dual codewords of weight $t^{\perp}n \approx \frac{k}{2}$, which is exactly the smallest weight achievable by classical polynomial-time algorithms, in particular with Prange's algorithm~\cite{Pra62}.

The work of~\cite{CT24} has the following limitations
\begin{enumerate}\setlength\itemsep{-0.2em}
	\item The $\alpha_{\iv}$ are of a very specific form, and the corresponding states $\ket{\psi_{\xv}}$ are product states.
	\item The measurements are restricted to learn $k$ independent bits of $\xv$ and not $k$ parities of $\xv$.
	\item Measurements used to unambiguously learn these $k$ bits are product measurements.
\end{enumerate}

These limitations raise the following questions
\begin{quote}
	Can we study Problem~\ref{Problem:AA} in all generality without any restriction on the amplitudes $\alpha_{\iv}$? Can we find better measurements for learning unambiguously $k$ parities of $\xv$ from $\ket{\psi_{\xv}}$ such that we can surpass Prange's barrier when used in~Proposition~\ref{Proposition:2}?
\end{quote}
We lacked a theoretical framework to do so, and the main purpose of this article is to fill this gap.

\subsection{Brief overview of contributions}
In this work, we introduce the notion of fine-grained unambiguous measurements. Instead of asking to recover the entire $\xv \in \F_2^n$ from $\ket{\psi_\xv}$ or output $\bot$, we ask whether it is possible to recover unambiguously some partial information about $\xv = (x_1,\dots,x_n)$. One could naturally consider learning unambiguously a subset $\{x_i\}_{i \in I}$ of the bits of $\xv$. This can be generalized to learning different parities $(\hv_1 \cdot \xv, \dots, \hv_k \cdot \xv)$. 

We formalize this by saying that our measurement outputs a pair $(\Hm,\yv)$, where $\Hm \in \F_2^{k \times n}$ is of full rank (for some $k \in \Iint{0}{n}$) and $\yv \in \F_2^k$. The output $(\Hm,\yv)$ corresponds to the information: ``I know with certainty that $\Hm\xv = \yv$’’ or equivalently $(\hv_1 \cdot \xv = \yv_1, \dots , \hv_k \cdot \xv = \yv_k)$, where the $\hv_i$ are the lines of $\Hm$. The $\bot$ outcome corresponds to the empty matrix/vector pair. Notice that this framework encompasses standard unambiguous measurements, where we only allow full rank matrices $\Hm \in \F_2^{n \times n}$ (which allows us to recover $\xv$ from $(\Hm,\yv)$ using Gaussian elimination) and empty matrices which corresponds to the $\bot$ outcome. 

As in the work of Chefles and Barnett, and motivated by the application to $\SLPN$, we restrict ourselves to sets of symmetric states $S = \{\ket{\psi_{\xv}}\}_{\xv \in \F_2^n}$. Our goal is, for a certain $S$, to bound the maximal number of parities that a fine-grained unambiguous measurement can correctly output given a random $\ket{\psi_\xv} \in S$. Here is a brief overview of our contributions.

\begin{enumerate}
	\item We formally introduce the notion of fine-grained unambiguous measurements. We give several theoretical results regarding fine-grained measurements. In particular, we show how the optimization that provides the best fine-grained unambiguous measurement for a symmetric set of states $S$ can be phrased as a linear program. 
	\item We study the associated dual linear program to derive upper bounds on the  number of parities. We provide specific solutions and also discuss the optimal bound. Our results generalize existing bounds for full unambiguous state discrimination. We apply our results to $\SLPN$. We show that one \emph{cannot} break the $\frac{k}{2}$ barrier of Prange’s algorithm using fine-grained unambiguous measurements. This shows the inherent limitations of the approach of~\cite{CT24}. We discuss the implications for the $\SLPN$ problem and, more generally, for state discrimination with prior information.
	\item We also discuss the computational efficiency of these measurements. We give sufficient conditions under which we can efficiently compute these fine-grained measurements. 
\end{enumerate}

\subsection{Detailed overview of contributions}
We start from a set of symmetric states $S = \{\ket{\psi_\xv}\}$. We fix an integer $n$. For $k \in \Iint{0}{n}$, let $\Lambda_k$ be the set of matrices in $\F_2^{k \times n}$ of full rank $k$. We first define the set $\Gamma(S)$ of fine-grained unambiguous measurements associated to $S$.
\begin{definition}
	Let $S = \{\ket{\psi_{\xv}}\}_{\xv \in \F_2^n}$ be a set of states. Let $\Gamma(S)$ be the set of measurements $\FHy$ satisfying
	\begin{enumerate}\setlength\itemsep{-0.2em}
		\item $\forall k \in \Iint{0}{n}, \ \forall (\Hm,\yv) \in \Lambda_k \times \F_2^k, \ F_{\Hm,\yv} \succeq \zerov $
		\item $  \sum_{k = 0}^n \sum_{\substack{\Hm \in \Lambda_k \\ \yv \in \F_2^k}} F_{\Hm,\yv} = I$.
		\item $	\forall k \in \Iint{0}{n}, \forall (\Hm,\yv) \in \Lambda_k \times \F_2^k, \ \forall \xv \ s.t. \ \Hm \xv \neq \yv, \ \Tr(F_{\Hm,\yv}\kb{\psi_\xv}) = 0.$
	\end{enumerate}
	The first two conditions ensure that $\FHy$ is a valid POVM and the last condition is the unambiguity condition, which means that the outcome $(\Hm,\yv)$ corresponds to the statement ``I know with certainty that $\Hm \xv = \yv$".
\end{definition}

Our goal is to upper bound the number of learned parities, which is given by the quantity $\rho(S) $ below
\begin{align*}\rho(S) & \eqdef \max_{\FHy \in \Gamma(S)} \rho(S,\FHy) \\ \text{ with } \quad  \rho(S,\FHy) & \eqdef \E_{\xv \leftarrow \F_2^n} \left[\sum_{k \in \Iint{0}{n}} \sum_{\substack{\Hm \in \Lambda_k \\ \yv \in \F_2^k}}  C(k) \Tr(F_{\Hm,\yv}\kb{\psi_\xv})\right],
\end{align*} $ \text{for some function } C : \Iint{0}{n} \rightarrow \mathbb{R}_+$. The ``score" $\rho(S)$ depends on this function $C$ as there are different ways of quantifying the quality of the best fine-grained unambiguous measurement. In this work, we will consider two settings:
\begin{itemize}
	\item The threshold setting. We are given a certain $\tau \in \Iint{1}{n}$ and we want to determine the maximum probability of learning at least $\tau$ parities. This is characterized by the function $C(k) = 1$ if $k \ge \tau$ and $C(k) = 0$ otherwise. 
	\item The average number of parities setting. We want to determine the maximum average number of parities learned. This is characterized by the function $C(k) = k$.
\end{itemize} 
Our focus is on these two scenarios, but our general results will apply for any function $C : \Iint{0}{n} \rightarrow \mathbb{R}_+$. 

\subsubsection{Expressing $\rho(S)$ as a linear program}

Upper bounding $\rho(S)$ and in particular maximizing over measurements that satisfy the fine-grained unambiguity condition seems hard to handle at first glance but we provide several simplifications that will make the problem easier, using the fact that the states $\{\ket{\psi_{\xv}}\}$ are symmetric. Ultimately, we show that $\rho(S)$ can be expressed as a maximization linear program, which is much easier to handle. This requires several careful steps. 

\paragraph{1. Simplifying the expression of $\rho(S)$ using the fact that $S$ is a set of symmetric states.}

We introduce the set of symmetric fine-grained measurements, which are defined as follows:

\begin{definition}\label{Definition:Gammas}
	Let $S = \{\ket{\psi_\xv}\}_{\xv \in \F_2^n}$ be a set of states. We define
	$$ \Gamma_s(S) = \left\{ \FHy \in \Gamma(S) : \forall (\Hm,\yv) \in \Lambda_k \times \F_2^k, \ \forall \av \in \F_2^n, \ X_{\av}F_{\Hm,\yv}X_{\av} = F_{\Hm,\yv + \Hm \av}\right\},$$
	where $X_\av$ is the Pauli shift operator in $\F_2^n$ satisfying $X_\av \ket{\xv} = \ket{\xv + \av}$.
\end{definition}

Our first result is to prove the following

\begin{restatable}{theorem}{thmqpu}\label{Theorem:4.1}
	Let $S = \{\ket{\psi_{\xv}}\}_{\xv \in \F_2^n}$ be a set of symmetric states with
	$ \ket{\psi_{\zerov}} = \sum_{\iv \in \F_2^n} \halpha_\iv \ket{\hi}.$
	Then
	$$ \rho(S) = \max \left\{
	\sum_{k \in \Iint{0}{n}}\sum_{\substack{\Hm \in \Lambda_k \\ \yv \in \F_2^k}}  \sum_{\iv \in \F_2^n} C(k)|\halpha_\iv|^2 \triple{\hi}{F_{\Hm,\yv}}{\hi}
	: \{F_{\Hm,\yv}\} \in \Gamma_s(S) \right\}.$$
\end{restatable}

In this new expression, we restricted the set of fine-grained measurements that we maximize on to the set $\Gamma_s(S)$. More importantly, we replaced the quantity $\E_{\xv \leftarrow \F_2^n} \left[\Tr(F_{\Hm,\yv} \kb{\psi_\xv})\right]$ with the quantity $\sum_{\iv \in \F_2^n} |\halpha_\iv|^2 \triple{\hi}{F_{\Hm,\yv}}{\hi}$ that depends only on the Fourier amplitudes of $\ket{\psi_\zerov}$ as well as on the Fourier diagonal elements of each $F_{\Hm,\yv}$. To prove this, we use the fact that we work with a set $S$ of symmetric states, and we also exploit the symmetries of the set $\Gamma_s(S)$ we introduce.

\paragraph{2. Relation on the Fourier diagonal elements of each $F_{\Hm,\yv}$.}

The issue with the above expression is that we still have to maximize over (symmetric) fine-grained unambiguous measurements. An appealing approach would be to define some variables $\lambda_\iv^{\Hm,\yv} = \triple{\hi}{F_{\Hm,\yv}}{\hi}$ and try to translate the condition $\FHy \in \Gamma_s(S) \subseteq \Gamma(S)$ into linear conditions on the $\lambda_\iv^{\Hm,\yv}$. This is actually possible to do and in order to present these linear relations, we have to introduce the notion of dual cosets.

\begin{definition}
	For each matrix $\Hm \in \Lambda_k$, we consider an arbitrary matrix $\Gm_\Hm \in \F_2^{n-k \times n}$ such that $Im(\trp{(\Gm_\Hm)}) = Ker(\Hm)$. We then define 
	$$ D_\Hm(\sv) = \{\xv \in \F_2^n : \Gm_\Hm \cdot \xv = \sv\}.$$
\end{definition}
Notice that this definition depends on the choice of $\Gm_\Hm$ but this choice only influences how the dual cosets are labeled. All our results will hold for any choice of $\Gm_\Hm$. We first restrict ourselves to sets $S = \{\ket{\psi_{\xv}}\}$ of symmetric states which have full dual support {\ie} $\forall \iv \in \F_2^n, \ \braket{\psi_\zerov}{\hi} \neq 0$, which implies $\forall \xv,\iv \in \F_2^n, \braket{\psi_\xv}{\hi} \neq 0$ from the fact that we have a symmetric set of states.
\noindent We prove the following
\begin{theorem}\label{Theorem:Intro2}
	Let $S = \{\ket{\psi_\xv}\}_{\xv \in \F_2^n}$ be a set of symmetric states with full dual support. We have 
	$$ \FHy \in \Gamma(S) \Rightarrow \forall k \in \Iint{0}{n}, \forall (\Hm,\yv) \in\Lambda_k \times \F_2^k, \ \forall \sv \in \F_2^{n-k}, \ \forall \iv,\jv \in D_\Hm(\sv), \ |\halpha_\iv|^2 \lambda_\iv^{\Hm,\yv} = |\halpha_\jv|^2 \lambda_\jv^{\Hm,\yv}.$$
\end{theorem}

In order to prove this statement, we look at the matrices $F_{\Hm,\yv}$. They are positive semidefinite and if we write $F_{\Hm,\yv} = \sum_i \mu_{\iv} \kb{A_i}$, the unambiguity condition tells us that we have  $\forall \xv \ s.t. \ \Hm\xv \neq \yv, \ \braket{A_i}{\psi_\xv} = 0$. We provide an explicit orthogonal basis of the space of states orthogonal to each $\ket{\psi_\xv}$ for $\xv$ such that $\Hm \xv = \yv$. The basis is expressed in terms of the dual cosets of $\Hm$ which then allows us to prove our theorem.

\paragraph{3. Rewriting $\rho(S)$ as a linear program}
We now plug in the relation from Theorem~\ref{Theorem:Intro2} into the expression of $\rho(S)$ from Theorem~\ref{Theorem:4.1}. We obtain the following expression 
\begin{definition}
	$$ \rho^L(S) \eqdef \max_{(\lambda_\iv^{\Hm,\yv})} \rho^L(S;(\lambda_\iv^{\Hm,\yv})) \eqdef \left\{\sum_{k = 0}^n \sum_{\substack{\Hm \in \Lambda_k \\ \yv \in \F_2^k}}\sum_{\iv \in \F_2^n} C(k) |\halpha_\iv|^2 \lambda_\iv^{\Hm,\yv} \right\},$$
	where the maximum is over nonnegative real numbers $(\lambda_\iv^{\Hm,\yv})$ satisfying 
	\begin{align}\label{Eq:I12}
		\forall \iv \in \F_2^n, \ \sum_{(\Hm,\yv)\in\I} \lambda_\iv^{\Hm,\yv} = 1
	\end{align}
	\begin{align}\label{Eq:I1} 
		\forall k \in \Iint{0}{n}, \forall (\Hm,\yv) \in\Lambda_k \times \F_2^k, \ \forall \sv \in \F_2^{n-k}, \ \forall \iv,\jv \in D_\Hm(\sv), \ |\halpha_\iv|^2 \lambda_\iv^{\Hm,\yv} = |\halpha_\jv|^2 \lambda_\jv^{\Hm,\yv}.
	\end{align}
\end{definition}
Because Theorem~\ref{Theorem:Intro2} is an implication, we have that $\rho(S) \le \rho^L(S)$ when $S$ has full dual support. In order to conclude, we have to deal with two issues:
\begin{enumerate}
	\item We have to show actually that $\rho(S) = \rho^L(S)$ when $S$ is a set of symmetric states with full dual support. To do so, we start from some $(\lambda_\iv^{\Hm,\yv})$ satisfying Equations~\ref{Eq:I1},~\ref{Eq:I12} and from these real numbers, we manage to construct a fine-grained POVM $\FHy$ such that $\rho(S;\FHy) \ge \rho^L(S;(\lambda_\iv^{\Hm,\yv}))$. 
	\item We have to remove the full dual support requirement. We show that every set $S$ can be approximated with another set of states that has full dual support and then use density arguments to show that if $\rho(S) = \rho^L(S)$ when $S$ has full dual support then this equality must also for any $S$ that doesn't have full dual support.
\end{enumerate}
Having dealt with these two final issues, we obtain the final theorem of this section
\begin{theorem}
	Let $S$ be a set of symmetric states. We have $\rho(S) = \rho^L(S)$.
\end{theorem}

\subsubsection{The dual linear program and upper bounds on $\rho(S)$}

\paragraph{1. Formulation of the dual linear program.}
Our goal is to provide upper bounds on $\rho(S)$. Since we have an expression of $\rho(S)$ as a linear optimization problem, it is natural to consider the associated dual linear program. We can show that the dual linear program can be expressed as follows

$$\sigma^L(S) = \min_{(b_\iv)_{\iv \in \F_2^n}} \sigma^L(S,(b_\iv)) \eqdef \sum_{\iv\in\F_2^n} b_\iv |\hat{\alpha}_\iv|^2,$$
where we minimize over nonnegative reals $(b_\iv)$ such that 
\begin{align}\label{Eq:I2}\forall k\in\Iint{0}{n}, \ \forall \Hm \in \Lambda_k, \ \forall \sv \in \F_2^{n-k}, \  \sum_{\iv\in \D_\Hm(\sv)} b_{\iv} \ge C(k) 2^k.\end{align}
We have directly by strong duality that $\rho^L(S) = \sigma^L(S)$, so finding nonnegative reals $(b_\iv)$ that satisfy Equation~\ref{Eq:I2} and computing $\sigma^L(S,(b_\iv))$ will yield an upper bound on $\rho(S) = \rho^L(S)$.

\paragraph{2. The threshold setting.} We first consider the threshold setting. Recall that we have a threshold $\tau \in \Iint{1}{n}$ and choose the function $C(k) = 1$ for $k \ge \tau$ and $C(k) = 0$ otherwise. Let $\rho(S,\tau),\rho^L(S,\tau)$ and $\sigma^L(S,\tau)$ be the (equal) values of the different optimization programs in this setting. Our first result is to prove a necessary and sufficient condition for which $\rho(S,\tau) \neq 0$ in this setting. This is a generalization of the result of Chefles and Barnett for regular unambiguous state discrimination. In order to do so, we have to introduce the notion of $\tau$-universal sets, which are subsets of $\F_2^n$ that intersect every affine subspace of $\F_2^n$ of dimension $\tau$. We then show

\begin{theorem}
	$$ \rho(S,\tau) = 0 \Leftrightarrow \text{There exists a } \tau\text{-universal set } V \ s.t. \ \forall \iv \in V, \ \halpha_\iv = 0.$$ 
\end{theorem}

For the $\Leftarrow$ implication, we actually have a stronger quantitative bound

\begin{proposition}\label{Proposition:I1}
	$$ \rho(S,\tau) \le \min\left\{2^{\tau} \sum_{\iv \in V} |\halpha_\iv|^2 : V \text{ is } \tau\text{-universal} \right\}.$$
\end{proposition}

For both sides of the implication, we use the dual formulation. In particular, for this last proposition, for any $\tau$-independent set $V$, we show that choosing $b_\iv = 2^\tau \one_V(\iv)$ satisfies the constraints of Equation~\ref{Eq:I2}, which implies that  $\sigma^L(S,\tau) \le \sigma^L(S,\tau;(b_\iv)) =  2^{\tau} \sum_{\iv \in V} |\halpha_\iv|^2$.

We are now ready to prove our main statement related to the original problem of solving $\SLPN$ using fine-grained unambiguous measurements. We define $B_d \eqdef \{\iv \in \F_2^n : |\iv|_H \le d\}$.

\begin{theorem}\label{Proposition:ThresholdBound}
	Let $S = \{\ket{\psi_{\xv}}\}_{\xv \in \F_2^n}$ be a set of symmetric states with $\ket{\psi_\zerov} = \sum_{\iv \in \F_2^n} \halpha_{\iv} \ket{\hi}$. Let $\gamma > 2$ be an absolute constant. Let $\eps$ such that $\sum_{\iv \notin B_d} |\halpha_\iv|^2 = \eps$. Then $\rho(S,\gamma d) \le \eps(1 + o(1))$, where $o(1)$ is a quantity that goes to $0$ as $d,n \rightarrow \infty$.
\end{theorem}

For states $\ket{\psi_\xv} = \sum_{\ev} f(\ev) \ket{\xv + \ev}$ where $f$ is a Bernoulli function of parameter $t$, if we write $\ket{\psi_\zerov} = \sum_{\iv \in \F_2^n} \halpha_\iv \ket{\hi}$, then 
$\sum_{\iv \notin B_d} |\halpha_\iv|^2 = negl(n)$ for $d = t^{\perp}(1+o(1))$. It was shown in~\cite{CT24} that one can learn around $2t^\perp$ coordinates of $\xv$ (so in particular parities of $\xv$) from $\ket{\psi_\xv}$. The above proposition shows that this is essentially optimal. 

In order to prove this proposition, we need a stronger statement than Proposition~\ref{Proposition:I1}. What we show using linear algebraic arguments is that $\overline{B_d}$ almost covers any affine subspace $V$ of $\F_2^n$ of dimension $\tau = \lceil \gamma d \rceil$, meaning that $\frac{|\overline{B_{d}} \cap V|}{|V|} = 1 - o(1)$. 

\paragraph{3. The average number of parities setting.}

We also study the average number of parities setting. Recall that here, we choose $C(k)= k$. Let $\rho_{Av}(S),\rho_{Av}^L(S),\sigma^L_{Av}(S)$ the values of the different optimization problems with this choice of $C$. 
We give $2$ families of dual solutions, and we give matching potential primal solutions. The first upper bound can be seen as an equivalent of Proposition~\ref{Proposition:ThresholdBound} and this shows that if the average weight of the dual support is $d$ then one can learn at most $2d$ parities of $\xv$ from $\ket{\psi_{\xv}}$.

\begin{theorem}\label{Theorem:6}
	$\rho_{Av}(S) \le 2\sum_{\iv \in \F_2^n} |\iv|_H |\halpha_{\iv}|^2$, where $|\cdot|_H$ is the Hamming weight of a binary vector. 
\end{theorem}
In order to prove this bound, we show using algebraic arguments that the choice $b_\iv = 2|\iv|_H$ is a valid solution of the dual linear program and we immediately have $\sigma_{Av}(S) \le \sigma_{Av}(S;(b_\iv)) = 2\sum_{\iv \in \F_2^n} |\iv|_H |\halpha_{\iv}|^2$. This proposition implies the following

\begin{corollary}
	Let $S = \{\ket{\psi_{\xv}}\}_{\xv \in \F_2^n}$ be a set of symmetric states with $\ket{\psi_\zerov} = \sum_{\iv \in \F_2^n} \halpha_{\iv} \ket{\hi}$. Assume that $\sum_{\iv \notin B_d} |\halpha_\iv|^2 = \eps$. Then
	$$ \rho_{Av}(S) \le 2d (1-\eps) + 2n\eps.$$
\end{corollary}

This means that if the dual support of $\ket{\psi_\zerov}$ is highly concentrated on words of weight at most $d$, then one cannot unambiguously learn on average much more than $2d$ parities of $\xv$ from $\ket{\psi_\xv}$.

The bound of Theorem~\ref{Theorem:6} is sometimes far from tight. We give another bound based on a different choice of dual solutions

\begin{proposition}
	$\rho_{Av}(S) \le (2^n+n-1)|\halpha_\zerov|^2 + (n-1) \sum_{\iv \in \F_2^n \setminus \{\zerov\}} |\halpha_\iv|^2.$
\end{proposition}

We also provide in Appendix~\ref{Appendix:n=2} a more in-depth study of the case $ n = 2$ where these two families of solutions span all possible dual solutions (which is not the case as $n$ increases). 

\subsubsection{Computational hardness}

Finally, we study the computational hardness of implementing such measurements. We managed to express $\rho(S)$ as a linear maximization program. A natural question is, given solutions $(\lambda_\iv^{\Hm,\yv})$ of the primal problem, whether one can efficiently construct fine-grained measurements from these solutions. We positively answer this question if we find a solution which is quantum sampleable.

\begin{theorem}\label{thm:cosntruction}
	Let $S$ be a set of symmetric states and let $(\lambda_\iv^{\Hm,\yv})$ be a primal solution, {\ie} an ensemble of nonnegative reals satisfying Equation~\ref{Eq:I12} and~\ref{Eq:I1}. Assume that these numbers are efficiently quantum sampleable {\ie} that the unitary 
	
	\begin{equation}\label{eq:impl_U}
		U : \ket{\iv}\ket{0} \mapsto \sum_{k \in \Iint{0}{n}} \sum_{\Hm \in \Lambda_k} \sqrt{\lambda_\iv^\Hm} \frac{|\halpha_\iv|}{\halpha_\iv} \ket{\iv}\ket{\Hm},
	\end{equation}
	can be computed in time $poly(n)$. Then we can construct a POVM $\FHy \in \Gamma(S)$ such that 
	\begin{enumerate}
		\item $\rho(S;\FHy) = \rho^L(S;(\lambda_\iv^\Hm))$.
		\item The POVM $\FHy$ can be efficiently implemented in time $poly(n)$.
	\end{enumerate}
\end{theorem}



\subsection{Discussion and Perspectives}

Let us put these results in perspective.

The main contribution of this work is to introduce fine-grained measurements and to perform an extensive study of these measurements. In the case of symmetric states, we provide a full characterization of the effectiveness of these measurements in terms of a linear program. This characterization allows us in particular to provide necessary and sufficient conditions for the existence of fine-grained unambiguous measurements that succeed with non zero probability, and to give precise upper bounds both in the threshold setting and in the average case setting. 

 Regarding our original question related to $\SLPN$, Theorem~\ref{Proposition:ThresholdBound} and, to some extent Theorem~\ref{Theorem:6} show that one cannot beat the $\frac{k}{2}$ barrier that we discussed in Section~\ref{Section:1.1}. However the optimal measurement for $\SLPN$, that uses the Pretty Good Measurement actually breaks this barrier. The way to interpret these two results is that there are actually two variants of the discrimination problem: 
\begin{enumerate}
	\item Given $\ket{\psi_{\cv}}$ for $\cv \in \C$ where $\C$ is a $k$-dimensional code, recover $\cv$. This is the actual $\SLPN$ problem. 
	\item Given $\ket{\psi_{\xv}}$ for $\xv \in \F_2^n$, recover $k$ parities of $\xv$. This is the problem we consider in our work. We then use the information that $\xv \in \C$ to recover $\xv$ from these $k$ parities.
\end{enumerate}

The second task is harder than the first one but it allows to work with a discrimination problem independently of the chosen code $\C$. The work of~\cite{CT24} gave information theoretic bounds for the first problem and we give in this work information theoretic bounds for the second problem. Our main conclusion is that the first task can be significantly easier than the second one, and that one cannot expect to beat Prange's bound if we only consider the second variant.

In other words, the main conclusion of this work is therefore that we have to strongly use the structure of the code $\C$ in order to solve the $\SLPN$, way beyond just recovering a codeword from $k$ parities. More generally, we introduced a framework for learning unambiguously $k$-parities of $\xv$ from a pure state $\ket{\psi_\xv}$. This is a very natural problem and we provided a extensive study of this question, at least for symmetric states. It would be interesting also to see what results hold if we remove this symmetry requirement. Finally, we believe fine-grained unambiguous measurements can have other applications, for example in the study of Quantum Oblivious Transfer or Quantum Random Access Codes.

\setcounter{theorem}{0}

\section{Preliminaries}
\paragraph{Notations.}
In this article, we work in the vector space $\F_2^n$. Vectors of $\F_2^n$ will be \emph{column vectors} and will be written in small bolds letters $\xv,\yv, \dots $. The canonical inner product in $\F_2^n$ between vectors $\xv$ and $\yv$ is denoted $\xv \cdot \yv$. Matrices will be written in capital bold letters $\Gm,\Hm,\dots$. We write $\trp{\xv}$ (resp. $\trp{\Mm}$) to denote the transpose of a vector (resp. of a matrix). For a Hermitian matrix $\Mm$, we write $\Mm \succeq 0$ when $\Mm$ is positive semidefinite.

For $\Mm \in \mathbb{C}^{\F_2^n \times \F_2^n}$, we write $(\Mm)_{\iv,\jv}$ to denote the entry of $\Mm$ in row $\iv$ and column $\jv$. In the \emph{bra-ket} notation, we can write $(\Mm)_{\iv,\jv} = \triple{\iv}{\Mm}{\jv}$.

\subsection{Quantum computing preliminaries}
We consider the Hadamard unitary $H : \ket{b} \rightarrow \frac{1}{\sqrt{2}}\left(\ket{0} + (-1)^{b}\ket{b}\right)$ which satisfies $H = H^\dagger$. The Quantum Fourier Transform over $\F_2^n$ is just the operation $H^{\otimes n}$.
\begin{definition}[Shift and Phase operators]
For $\bv$ in $\mathbb{F}_2^n$, let $X_{\bv}$ be the shift operator $X_{\bv}\ket{\xv} = \ket{\xv + \bv}$ and $Z_{\bv}$ be the phase operator $Z_{\bv}\ket{\xv} = (-1)^{\xv \cdot \bv}$$ \ket{\xv}$.
\end{definition}
Notice that because we work in $\F_2^n$, we have $X^\dagger_\bv = X_\bv$.

\begin{definition}\label{def:action_qft}
	For a vector $\xv \in \F_2^n$, we define $\ket{\hx} = H^{\otimes n} \ket{\xv} = \frac{1}{\sqrt{2^n}}\sum_{\yv \in \F_2^n} (-1)^{\xv\cdot\yv} \ket{\yv} $. For a matrix $\Mm \in \mathbb{C}^{\F_2^{n} \times \F_2^n}$, we define $\widehat{\Mm} = H^{\otimes n} \Mm H^{\otimes n}$.
\end{definition}

Let $\Mm \in \mathbb{C}^{\F_2^n \times \F_2^n}$. From our definitions, we have that for any $\xv,\yv \in \F_2^n$, we have $\triple{\xv}{\widehat{\Mm}}{\yv} = \triple{\hx}{\Mm}{\hy}$. The elements $\triple{\hx}{\Mm}{\hx}$ are called the Fourier diagonal elements of $\Mm$.

\begin{proposition}\label{proposition:shift}
	We have for all $\bv$ in $\mathbb{F}_2^n$ that 
	$ \ket{\hx}$  is  an eigenstate of $X_\bv$ associated to the eigenvalue  $(-1)^{\xv \cdot \bv}$ and
	\begin{eqnarray}
	X_\bv \cdot H^{\otimes n} & = & H^{\otimes n} \cdot Z_{\bv} \label{eq:eigenstate}\\
	H^{\otimes n} \cdot X_\bv &=& Z_\bv \cdot H^{\otimes n}. \label{eq:dual_eigenstate}
	\end{eqnarray}
\end{proposition}

\subsection{Binary linear codes}
Let $\P_2(n)$ be the set of all subspaces of $\F_2^n$ and let $\G_2(n,k)$ be the set of all subspaces of $\F_2^n$ of dimension $k$. In particular, we have $\P_2(n) = \bigcup_{k = 0}^n \G_2(n,k)$.

A binary linear code $\C$ of dimension $k$ and length $n$ is an element of $\G_2(n,k)$. It can be characterized both by a generator matrix $\Gm \in \F_2^{k \times n}$ of rank $k$ or a parity matrix $\Hm \in \F_2^{(n-k) \times n}$ of rank $(n-k)$, so that 
$$
	\C = \{\trp{\Gm}\sv : \sv \in \F_2^k\} = \{\cv \in \F_2^n : \Hm \cv = \zerov\} = Ker(\Hm),$$
where again, we use the convention that $\C$ is a column subspace. Notice that the same code $\C \in \G_2(n,k)$ can have many different generator and parity matrices. The dual code of $\C$ is $\C^\bot = \{\yv \in \F_2^n : \forall \cv \in \C, \ \cv \cdot \yv = 0\}$. One can check that a generator matrix $\Gm$ of $\C$ is also a parity matrix of $\C^\bot$.

\begin{definition}
	We denote $\Lambda_{n,k}$ the set of matrices $\Mm \in \F_2^{k \times n}$ of full rank $k$. When $n$ is clear from context, we will omit the dependency in $n$ and write $\Lambda_k$ instead of $\Lambda_{n,k}$.
\end{definition}

It will be useful to fix a generator matrix associated to a code.
\begin{definition}
	$\forall k \in \Iint{0}{n}$, $\forall \C \in \G_2(n,k)$, we associate to $\C$ a fixed arbitrary matrix $\Gm_\C \in \Lambda_k$ such that $span\{\trp{(\Gm_\C)} \sv : \sv \in \F_2^k\} = \C$.
	If the code $\C$ is specified by a parity matrix $\Hm$, we also denote by $\Gm_\Hm$, this associated generator matrix.
\end{definition} 

\begin{definition}\label{Definition:DualCosets}
	Let $\Hm \in \Lambda_k$. Let $\C = Ker(\Hm) \in \G_2(n,n-k)$, which implies that $\Hm \trp{(\Gm_\C)} = \trp{\Hm}\Gm_\C = \zerov$. We define the dual cosets $\D_\Hm(\sv)$ associated to $\Hm$ as follows
	$$ \forall \sv \in \F_2^{n-k}, \ \D_\Hm(\sv) \eqdef \{\xv \in \F_2^n :   \Gm_\C \cdot \xv = \sv\}.$$
	Let any $\vv \in \F_2^n$ such that $\Gm_\C \cdot \vv = \sv$. Notice that we can also write 
	$$ \D_\Hm(\sv) = \{\trp{\Hm}\uv  + \vv : \uv \in \F_2^k\}.$$
\end{definition}

The following proposition will be very useful throughout this work. 

\begin{proposition}\label{Proposition:CodeSum}
	Let $\mathcal{C} \in \G_2(n,k)$ be a linear code. 
	\begin{equation*}
		\sum_{\cv \in\mathcal{C}} (-1)^{\vv \cdot \cv}
		= \left\{
		\begin{array}{ll}
			|\mathcal{C}| & \text{if } \vv \in \C^\bot \\
			0 & \text{otherwise}
		\end{array}
		\right.
	\end{equation*}
\end{proposition}

\begin{proof}
    If $\vv \in\mathcal{C}^\bot$, then $\vv \cdot \cv = 0$ for all $\cv \in\mathcal{C}$. If $\vv \notin\mathcal{C}^\bot$, then there exists $\cv_0\in\mathcal{C}$ such that, $\vv \cdot \cv_0 = 1$.
    \begin{equation*}
        \sum_{\cv\in\mathcal{C}} (-1)^{\vv\cdot\cv} = \sum_{\cv\in\mathcal{C}} (-1)^{\vv\cdot(\cv+\cv_0)} = (-1)^{\vv\cdot\cv_0} \sum_{\cv\in\mathcal{C}} (-1)^{\vv\cdot\cv} = - \sum_{\cv\in\mathcal{C}} (-1)^{\vv \cdot \cv}
    \end{equation*}
    Thus, $\sum_{\cv\in\mathcal{C}} (-1)^{\vv\cdot\cv} = 0$.
\end{proof}

\subsection{Fine-Grained Unambiguous Measurements on $\mathbb{F}_2^n$}
We want to capture what it means to learn $k$ parities of $\xv = \begin{pmatrix} x_1 \\ \vdots \\ x_n \end{pmatrix}$. A way to capture this is to consider matrices $\Hm \in \F_2^{k \times n}$ of rank $k$, and to write $\Hm \trsp{\xv} = \trsp{\yv}$. For example, if
$$ \Hm = \begin{pmatrix} 1 $ 0 $ 1 \\ 0 $ 1 $ 1 \end{pmatrix} \textrm{ and } \yv = \begin{pmatrix} 1 \\ 0 \end{pmatrix} \textrm{ then } \Hm \trsp{\xv} = \trsp{\yv} \Leftrightarrow x_1 \oplus x_3 = 1 \wedge x_2 \oplus x_3 = 0.$$
Here, we know $2$ parities of $\xv$. We now define the information sets associated to the knowledge of such parities:
\begin{definition}
	$$
	\forall k \in \Iint{0}{n}, \ \I_k \eqdef  \{(\Hm,\yv) \in \Lambda_k \times \F_2^k\} \quad ; \quad \I \eqdef \cup_{k = 0}^n \I_k \quad ; \quad  \I^* = \I\setminus\I_0.$$
\end{definition}
 We are now ready to define Fine-Grained Unambiguous Measurements over a set of states $S$.
\begin{definition}[Fine-Grained Unambiguous Measurements on $\mathbb{F}_2^n$]
	A fine-grained unambiguous measurement on $\F_2^n$ associated to a set of states $\{ \ket{\psi_{\xv}}, \xv \in\mathbb{F}_2^n \}$ is a POVM $\{F_{\Hm,\yv}\}_{(\Hm,\yv) \in \I}$  such that
	\begin{equation*}
		\forall (\Hm,\yv) \in \I, \ \forall \xv \in\mathbb{F}_2^n \text{ s.t. } \Hm\trsp{\xv} \ne \trsp{\yv}, \ \Tr(F_{\Hm,\yv}\kb{\psi_{\xv}}) = 0.
	\end{equation*}
	When $(\Hm,\yv)$ are the null matrix and null vector, the associated $POVM$ element is also denoted $F_\bot$. In particular, we can write 
	$$ F_{\bot} = I - \sum_{(\Hm,\yv) \in \I^*} F_{\Hm,\yv}.$$
\end{definition}

The different outcomes $(\Hm,\yv)$ correspond to the information: ``I know with certainty that $\Hm \trsp{\xv} = \trsp{\yv}$.'' 

\begin{definition}
	Given a set of states $S = \{\ket{\psi_{\xv}}\}_{\xv \in \F_2^n}$, the set of fine-grained unambiguous measurements associated to $S$ is denoted $\Gamma(S)$.
\end{definition}

\paragraph{The associated optimization program.}

It is simple to evaluate the quality of a standard unambiguous measurement, just by looking at the probability of success. In our setting, we will have a distribution over the number of parities we obtain. As we can think of, the two most natural settings are the threshold setting, where we ask the measurement to output at least $\tau$ parities and the average parity setting where we look at the average number of parities that the measurement outputs. Both these settings, and many other settings, can be encompassed by the quantity below

\begin{definition}
	Let $S = \{\ket{\psi_{\xv}}\}_{\xv \in \F_2^n}$ be a set of states and $\{F_{\Hm,\yv} : (\Hm,\yv) \in \I\} \in \Gamma(S)$. The quality of the measurement $\FHy$ on the set $S$ for some function $C : \Iint{0}{n} \rightarrow \mathbb{R}$ is given by the quantity:
	$$ \rho(S,\{F_{\Hm,\yv}\}) \eqdef \E_{\xv \in \F_2^n}\left[\sum_{k \in \Iint{0}{n}}\sum_{(\Hm,\yv) \in \I_k} C(k) \cdot \Tr\left(F_{\Hm,\yv} \kb{\psi_{\xv}}\right)\right].$$
	The optimal value of a measurement for this choice of function $C$ then becomes
	$$ \rho(S) \eqdef \max_{\{F_{\Hm,\yv}\} \in \Gamma(S)} \rho(S,\{F_{\Hm,\yv}\}).$$
\end{definition}

The threshold setting then corresponds to the case where $C(k) = 1$ for $k \ge \tau$ and $C(k) = 0$ otherwise. The average parity setting corresponds to the case $C(k) = k$.

\paragraph{Symmetric states.}
In this article, we will work on sets of symmetric states.
\begin{definition}
	A set $S = \{\ket{\psi_{\xv}}\}_{\xv \in \F_2^n}$ is called symmetric iff. $\forall \xv \in \F_2^n, \ \ket{\psi_\xv} = X_\xv \ket{\psi_\zerov}$.
\end{definition}

If we write $\ket{\psi_\zerov} = \sum_{\iv \in \F_2^n} \halpha_\iv \ket{\hi}$, we have  $\ket{\psi_\xv} = \sum_{\iv \in \F_2^n} \halpha_\iv (-1)^{\xv \cdot \iv}\ket{\hi}$.

\section{Reformulation as a linear program}
In this section, our goal is to show that the optimization program for $\rho(S)$ can be rephrased in terms of a linear program when $S$ is a set of symmetric states. Our proof will go in three steps:
\begin{enumerate}
	\item We first provide a simplified expression of $\rho(S)$. We introduce the notion of symmetrized measurements which allows us in particular to express the objective $\rho(S)$ as a function of the diagonal Fourier coefficients of the matrices $\{F_{\Hm,\yv}\}$. 
	\item We then show linear relations between the different diagonal Fourier coefficients of the matrices $\{F_{\Hm,\yv}\}$ using the unambiguity condition of these measurements. 
	\item We use these relations to relax the optimization program into a linear program with objective $\rho^L(S)$. Finally, we show that this relaxation doesn't change the objective value, namely that $\rho(S) = \rho^L(S)$.
\end{enumerate}

\subsection{Reformulation of $\rho(S)$ involving the Fourier diagonal elements of $F_{\Hm,\yv}$}\label{Sec:3.1}

We introduce the set of symmetric fine-grained measurements, which are defined as follows:

\begin{definition}\label{Definition:Gammas}
	Let $S$ be a set of states. We define
	$$ \Gamma_s(S) = \left\{ \FHy \in \Gamma(S) : \forall (\Hm,\yv) \in \I, \ \forall \av \in \F_2^n, \ X_{\av}F_{\Hm,\yv}X_{\av} = F_{\Hm,\yv + \Hm \av}\right\}.$$
\end{definition}

\noindent The main goal of this section is to prove the following theorem.

\begin{restatable}{theorem}{thmqpu}\label{Theorem:4.1}
	Let $S = \{\ket{\psi_{\xv}}\}_{\xv \in \F_2^n}$ be a set of symmetric states with
	$ \ket{\psi_{\zerov}} = \sum_{\iv \in \F_2^n} \halpha_\iv \ket{\hi}.$
	Then
	$$ \rho(S) = \max \left\{
	\sum_{k \in \Iint{0}{n}}\sum_{(\Hm,\yv) \in \I_k}  \sum_{\iv \in \F_2^n} C(k)|\halpha_\iv|^2 \triple{\hi}{F_{\Hm,\yv}}{\hi}
	: \{F_{\Hm,\yv}\} \in \Gamma_s(S) \right\}.$$
\end{restatable}

In order to prove this theorem, we first present a  simplification of $\rho(S)$, given by the proposition below. 

\subsubsection{First simplification}
\begin{proposition}\label{Proposition:4.1}
	Let $S = \{\ket{\psi_{\xv}}\}$ be a set of symmetric states. For any $\{F_{\Hm,\yv}\}$, we define 
	$$ \rho_2(S,\{F_{\Hm,\yv}\}) \eqdef \sum_{k \in \Iint{0}{n}}\sum_{(\Hm,\yv) \in \I_k} C(k) \cdot \Tr\left(F_{\Hm,\yv} \kb{\psi_{\zerov}}\right).$$ Then
	$ \rho(S) = \max\left\{\rho_2(S,\{F_{\Hm,\yv}\}) : \{F_{\Hm,\yv}\} \in \Gamma_s(S)\right\}$.
\end{proposition}
\begin{proof}
	We fix a set $S$ of symmetric states. We start with the following lemma
	\begin{lemma}\label{Lemma:4.1.1}
		$\forall \{F_{\Hm,\yv}\} \in \Gamma_s(S), \ \rho(S,\{F_{\Hm,\yv}\}) = \rho_2(S,\FHy)$.
	\end{lemma}
	\begin{proof}
		We fix any $\{F_{\Hm,\yv}\} \in \Gamma_s(S)$ and write
		\begin{align*}
			\rho(S,\{F_{\Hm,\yv}\}) & = \frac{1}{2^n} \sum_{\xv \in \F_2^n} \sum_{k \in \Iint{0}{n}} C(k) \sum_{(\Hm,\yv) \in \I_k} \Tr(F_{\Hm,\yv} \kb{\psi_{\xv}}) \\
			& = \frac{1}{2^n} \sum_{\xv \in \F_2^n} \sum_{k \in \Iint{0}{n}} C(k) \sum_{(\Hm,\yv) \in \I_k} \Tr(X_{\xv}F_{\Hm,\yv} X_{\xv}\kb{\psi_{\zerov}}) & \text{since } S \text{ is a set of symmetric states} \\
			& = \frac{1}{2^n} \sum_{\xv \in \F_2^n} \sum_{k \in \Iint{0}{n}} C(k) \sum_{(\Hm,\yv) \in \I_k} \Tr(F_{\Hm,\yv + \Hm\xv} \kb{\psi_{\zerov}}) & \text{since } \FHy \in \Gamma_s(S) \\
			& = \frac{1}{2^n} \sum_{k \in \Iint{0}{n}} C(k) 2^{n-k} \sum_{(\Hm,\yv) \in \I_k} \sum_{\av \in \F_2^k} \Tr(F_{\Hm,\av} \kb{\psi_\zerov}) & \text{introducing new variable } \av = \yv + \Hm\xv \\
			& = \sum_{k \in \Iint{0}{n}} C(k) \sum_{(\Hm,\av) \in \I_k} \Tr(F_{\Hm,\av} \kb{\psi_\zerov}) \\
			& = \rho_2(S,\{F_{\Hm,\yv}\})
		\end{align*}
	\end{proof}
	\begin{lemma}\label{Lemma:4.1.2}
		Let $\{F_{\Hm,\yv}\} \in \Gamma(S)$ and let  $\{\bF_{\Hm,\yv}\}$ such that 
		$$ \forall (\Hm,\yv) \in \I, \ \bF_{\Hm,\yv} \eqdef \frac{1}{2^{n}}\sum_{\av \in \F_2^n} X_{\av} F_{\Hm,\yv + {\Hm}\av}X_{\av}.$$
		Then $\{\bF_{\Hm,\yv}\} \in \Gamma_s(S)$.
	\end{lemma}
	\begin{proof}
		Regarding non-negativity, we write 
		$$ \forall (\Hm,\yv) \in \I, \ \bF_{\Hm,\yv} = \frac{1}{2^n}\sum_{\av \in \F_2^n} X_{\av} F_{\Hm,\yv + \Hm\av}X_{\av}.$$
		Now, since $\forall \av \in \F_2^n, \ X_{\av} = X_{\av}^\dagger$ and each $F_{\Hm,\yv + \Hm\av} \succeq \zerov$, we directly have that each $ X_{\av} F_{\Hm,\yv + \Hm\av}X_{\av} \succeq 0$ hence $\bF_{\Hm,\yv} \succeq \zerov$. 
		
		Now fix any $(\Hm,\yv) \in \I$ and $\xv \in \F_2^n$ such that $\Hm \xv \neq \yv$. We write 
		\begin{align*}
			\Tr(\bF_{\Hm,\yv}\kb{\psi_{\xv}}) & = \frac{1}{2^n} \sum_{\av \in \F_2^n} \Tr(X_\av F_{\Hm,\yv + \Hm\av} X_\av \kb{\psi_{\xv}}) \\
			& = \frac{1}{2^n} \sum_{\av \in \F_2^n} \Tr(F_{\Hm,\yv + \Hm\av} X_\av \kb{\psi_{\xv}}X^\dagger_\av) \\
			& = \frac{1}{2^n} \sum_{\av \in \F_2^n} \Tr(F_{\Hm,\yv + \Hm\av} \kb{\psi_{\xv+\av}}) \\
			& = 0
		\end{align*}
		where in the last equality, we use that fact that $\Hm\xv \neq \yv \Rightarrow \forall \av \in \F_2^n, \ \Hm(\xv + \av) \neq \yv + \Hm \av$ and the fact that $\{F_{\Hm,\yv}\} \in \Gamma(S)$. Finally, we have 
		\begin{align*}
			\sum_{(\Hm,\yv) \in \I} \bF_{\Hm,\yv} & = \frac{1}{2^n} \sum_{\av \in \F_2^n} \sum_{(\Hm,\yv) \in \I} X_\av F_{\Hm,\yv} X_\av \\
			& = \frac{1}{2^n} \sum_{\av \in \F_2^n}  X_\av I X_\av \\
			& = I
		\end{align*}
		We proved that $\{\bF_{\Hm,\yv}\} \in \Gamma(S)$. For the second requirement of Definition~\ref{Definition:Gammas}, we write
		\begin{align*} X_{\bv}\bF_{\Hm,\yv}X_{\bv} & = \frac{1}{2^{n}} \sum_{\av \in \F_2^n} X_{\av + \bv} F_{\Hm,\yv + \Hm \av} X_{\av + \bv} \\
			& = \frac{1}{2^{n}}  \sum_{\av' \in \F_2^n} X_{\av'} F_{\Hm,\yv + \Hm\bv + \Hm\av'} X_{\av'}  & \text{variable change } \av' = \av + \bv \\ & =  \bF_{\Hm,\yv + \Hm\bv}
		\end{align*}
		which concludes the proof of our lemma.
	\end{proof}
	
	We can now continue the proof of Proposition~\ref{Proposition:4.1}, and we will prove the equality by proving the inequality both ways. First, we write 
	\begin{align*}
		\rho(S) & = \max\left\{\rho(S,\{F_{\Hm,\yv}\}) : \FHy \in \Gamma(S)\right\} \\
		& \ge \max\left\{\rho(S,F_{\Hm,\yv}) : \FHy \in \Gamma_s(S)\right\} & \text{since } \Gamma_s(S) \subseteq \Gamma(S) \\
		& = \max\left\{\rho_2(S,\FHy) : \FHy \in \Gamma_s(S)\right\} & \text{from Lemma}~\ref{Lemma:4.1.1} 
	\end{align*}
	For the reverse inequality, let $\{F^\MAX_{\Hm,\yv}\} \in \Gamma(S)$ such that $\rho(S) = \rho(S,\{F^\MAX_{\Hm,\yv}\})$, and let $\{\bF_{\Hm,\yv}\}$ such that  $$ \forall (\Hm,\yv) \in \I, \ \bF_{\Hm,\yv} \eqdef \frac{1}{2^{n}}\sum_{\av \in \F_2^n} X_{\av} F^\MAX_{(\Hm,\yv + {\Hm}\av)}X_{\av}.$$
	We now write 
	\begin{align*}
		\max\{\rho_2(S,\FHy) : \FHy \in \Gamma_s(S)\} & \ge
		\rho_2(S,\bF_{\Hm,\yv}) & \text{from Lemma}~\ref{Lemma:4.1.2} \\  
		& = \sum_{k \in \Iint{0}{n}} C(k) \sum_{(\Hm,\yv) \in \I_k} \Tr(\bF_{\Hm,\yv} \kb{\psi_{0}}) \\
		&  = \sum_{k \in \Iint{0}{n}} C(k) \sum_{(\Hm,\yv) \in \I_k} \frac{1}{2^n} \sum_{\av \in \F^2_n} \Tr(X_\av F^\MAX_{\Hm,\yv + \Hm\av} X_{\av}\kb{\psi_{0}}) \\
		& = \frac{1}{2^n} \sum_{\av \in \F_2^n}\sum_{k \in \Iint{0}{n}} C(k) \sum_{(\Hm,\yv) \in \I_k} \Tr(F^\MAX_{\Hm,\yv + \Hm\av} \kb{\psi_{\av}}) \\
		& = \frac{1}{2^n} \sum_{\av \in \F_2^n}\sum_{k \in \Iint{0}{n}} C(k) \sum_{(\Hm,\yv') \in \I_k} \Tr(F^\MAX_{\Hm,\yv'} \kb{\psi_{\av}}) & (\yv' = \yv + \Hm\av)\\
		& = \rho(S,\{F^\MAX_{\Hm,\yv}\}) \\
		& = \rho(S)
	\end{align*}
	Putting our two inequalities together, we obtain 
	$$ \rho(S) = \max\left\{\rho_2(S,\FHy) : \FHy \in \Gamma_s(S)\right\}.$$
\end{proof}
\subsubsection{Fourier diagonal terms}
\begin{proposition}\label{Proposition:DiagonalFH}
	Let $S = \{\ket{\psi_\xv}\}_{\xv \in \F_2^n}$ be a set of symmetric states and let $\{F_{\Hm,\yv}\} \in \Gamma_s(S)$. We define
	$$ F_\Hm \eqdef \sum_{\yv \in \F_2^k} F_{\Hm,\yv} \quad \forall k \in \Iint{0}{n}, \ \forall \Hm \in \Lambda_k.$$
	Then for each $(\Hm,\yv) \in \I$ and $\iv,\jv \in \F_2^n$, we have 
	$$ (\hFH)_{\iv,\jv} = \left\{\begin{tabular}{cl}
		$0 $& { if } $\iv \neq \jv$ \\
		$2^k \bra{\iv} \widehat{F_{\Hm,\yv}} \ket{\iv}$ & { if } $\iv = \jv $
	\end{tabular}\right.
	$$
	In particular, each $\widehat{F_\Hm}$ is a diagonal matrix, and the diagonal terms $\bra{\iv} \widehat{F_{\Hm,\yv}} \ket{\iv}$ are independent of $\yv$.
\end{proposition}
\begin{proof}
	We fix any $(\Hm,\yv) \in \I$. Let any $\bv \in \F_2^n$ such that $\Hm \bv = \yv$. We write 
	\begin{align*}
		\hFH & = H^{\otimes n} F_\Hm H^{\otimes n} =  H^{\otimes n} \sum_{\yv \in \F_2^k} F_{\Hm,\yv} H^{\otimes n} = \frac{1}{2^{n-k}}\sum_{\av \in \F_2^{n}} H^{\otimes n} F_{\Hm,\Hm\av} H^{\otimes n} \\
		& = \frac{1}{2^{n-k}}\sum_{\av \in \F_2^n} H^{\otimes n}  X_{\av + \bv} F_{\Hm,\yv} X_{\av + \bv}  H^{\otimes n} = \frac{1}{2^{n-k}}\sum_{\av \in \F_2^n} Z_{\av + \bv} \widehat{F_{\Hm,\yv}} Z_{\av + \bv},
	\end{align*}
	where we used $F_{\Hm,\Hm\av} = X_{\av + \bv} F_{\Hm,\Hm\bv}X_{\av + \bv}$ since $\FHy \in \Gamma_s(S)$.
	Now, fix any $\iv,\jv \in \F_2^n$. We write
	\begin{align*}
		(\hFH)_{\iv,\jv} & = \bra{\iv} \hFH \ket{\jv} =  \bra{\iv} \frac{1}{2^{n-k}}\sum_{\av \in \F_2^n} Z_{\av + \bv} \widehat{F_{\Hm,\yv}} Z_{\av + \bv}  \ket{\jv} = \frac{1}{2^{n-k}} \sum_{a \in \F_2^n} (-1)^{(\av + \bv) \cdot (\iv + \jv)} \bra{\iv} \widehat{F_{\Hm,\yv}} \ket{\jv} \\
		& = \left\{\begin{tabular}{cl}
			$0 $& { if } $\iv \neq \jv$ \\
			$2^k \bra{\iv} \widehat{F_{\Hm,\yv}} \ket{\iv}$ & { if } $\iv = \jv $
		\end{tabular}\right. \qedhere
	\end{align*}
\end{proof}
\subsubsection{Proof of Theorem~\ref{Theorem:4.1}}
We can now conclude the proof of our theorem, which we restate below
\thmqpu*

\begin{proof}
	We fix a set $S$ of symmetric states. From Proposition~\ref{Proposition:4.1}, we have that \begin{align*}
		\rho(S) & = \max\left\{\rho_2(S,\FHy) : \FHy \in \Gamma_s(S)\right\} \\
		\text{with} \quad  \rho_2(S,\{F_{\Hm,\yv}\}) & = \sum_{k \in \Iint{1}{n}}\sum_{(\Hm,\yv) \in \I_k} C(k) \cdot \Tr\left(F_{\Hm,\yv} \kb{\psi_{\zerov}}\right).
	\end{align*}
	Now fix any $\FHy \in \Gamma_s(S)$ and, for each $(\Hm,\yv) \in \I_k, $ we define $F_\Hm = \sum_{\yv \in \F_2^k} F_{\Hm,\yv}$. We now write
	\begin{align*}
		\rho_2(S,\{F_{\Hm,\yv}\}) & = \sum_{k \in \Iint{0}{n}}\sum_{(\Hm,\yv) \in \I_k} C(k) \cdot \Tr\left(F_{\Hm,\yv} \kb{\psi_{\zerov}}\right) \\
		& = \sum_{k \in \Iint{0}{n}}\sum_{\Hm \in \Lambda_k} C(k) \cdot \Tr\left(F_\Hm \kb{\psi_{\zerov}}\right) \\
		& = \sum_{k \in \Iint{0}{n}}\sum_{\Hm \in \Lambda_k}\sum_{\iv \in \F_2^n}  C(k) |\halpha_\iv|^2 \triple{\hi}{F_\Hm}{\hi} & \text{ from Proposition}~\ref{Proposition:DiagonalFH} \\
		& = \sum_{k \in \Iint{0}{n}}\sum_{(\Hm,\yv) \in \I_k}\sum_{\iv \in \F_2^n}  C(k) |\halpha_\iv|^2 \triple{\hi}{F_{\Hm,\yv}}{\hi} 
	\end{align*}
	Since $\rho(S) = \max\left\{\rho_2(S,\FHy) : \FHy \in \Gamma_s(S)\right\}$, we get our theorem. 
\end{proof}

\subsection{Relations between diagonal Fourier coefficients of the matrices $F_{\Hm,\yv}$.}
We managed to give an expression for $\rho(S)$ where we maximize a quantity over the set of symmetric measurements that depends only on the Fourier diagonal elements of each POVM element. Now, we show that the unambiguity condition of our measurement can be translated into relations between these different Fourier diagonal elements, which in turn will give yet another formulation for $\rho(S)$ as the maximum of a linear program.
\label{Sec:3.2}
\subsubsection{Full Dual Support}
In our proof, we will first have to restrict ourselves to symmetric states $S$ which have full dual support. We will then be able to prove our general bound for any symmetric $S$ using density arguments. 
\begin{definition}
	A set of symmetric states $S = \{\ket{\psi_\xv}\}_{\xv \in \F_2^n}$ has full dual support iff. 
	$$ \forall \yv \in \F_2^n, \braket{\hy}{\psi_{\zerov}} \neq 0.$$
\end{definition}
Notice that this implies that $\forall \xv,\yv \in \F_2^n, \ \braket{\hy}{\psi_{\xv}} \neq 0$, since from the symmetry condition, we have $\braket{\hy}{\psi_{\xv}} = (-1)^{\xv \cdot \yv} \braket{\hy}{\psi_{\zerov}}$. 
We first have the following proposition.
\begin{proposition}\label{Proposition:FullDual}
	Let $S = \{\ket{\psi_\xv}\}_{\xv \in \F_2^n}$ be a set of symmetric states with full dual support. We have 
	$$ dim\left(span\{\ket{\psi_\xv}\}_{\xv \in \F_2^n}\right) = 2^n.$$
	In other words, the $\{\ket{\psi_\xv}\}$ form a basis of $\F_2^n$. This implies that $\forall T \subseteq \F_2^n, \ dim\left(span\{\ket{\psi_\xv}\}_{\xv \in T}\right) = |T|.$ 
\end{proposition}
\begin{proof}
	We fix any $\yv_0 \in \F_2^n$ and show that $\ket{\hyz} \in span\{\ket{\psi_\xv}\}_{\xv \in \F_2^n}$ which will imply the desired statement. We write $\ket{\psi_\zerov}$ in the Fourier basis:
	$$ \ket{\psi_\zerov} = \sum_{\yv \in \F_2^n} \halpha_\yv \ket{\hy},$$
	where each $\halpha_{\yv} \neq 0$ since $S$ has full dual support. Now, we compute
	\begin{align*}
		\sum_{\xv \in \F_2^n} (-1)^{\xv \cdot \yv_0} \ket{\psi_{\xv}} & = \sum_{\xv \in \F_2^n} (-1)^{\xv \cdot \yv_0} \sum_{\yv \in \F_2^n} (-1)^{\xv \cdot \yv} \halpha_\yv \ket{\hy} \\
		& =  \sum_{\yv \in \F_2^n} \halpha_\yv \left(\sum_{\xv \in \F_2^n} (-1)^{\xv \cdot (\yv + \yv_0)}\right)  \ket{\hy} \\
		& = 2^n \halpha_{\yv_0} \ket{\hyz}. 
	\end{align*}
	Since $	\sum_{\xv \in \F_2^n} (-1)^{\xv \cdot \yv_0} \ket{\psi_{\xv}} \in span\{\ket{\psi_\xv}\}_{\xv \in \F_2^n}$ and $\halpha_{\yv_0} \neq 0$, we get that $\ket{\hyz} \in span\{\ket{\psi_\xv}\}_{\xv \in \F_2^n}$. Since this true for each $\yv_0$ and the $\{\ket{\hyz}\}$ form an orthonormal basis of $\F_2^n$, we get our result.
\end{proof}
\subsubsection{Characterizing the $\{F_{\Hm,\yv}\}$}
The goal of this section is to take advantage of the fact that we have an unambiguous measurement to characterize the matrices $F_{\Hm,\yv}$. In particular, we want to give properties on the diagonal elements of $\widehat{F_{\Hm,\yv}}$.
\begin{definition}
	Let $S = \SS$ be a set of symmetric states. We define 
	$$ \forall (\Hm,\yv) \in \I, \quad  V_{\Hm,\yv}^{\neq} \eqdef span\left\{\ket{\psi_\xv} : \Hm\xv \neq \yv \right\} \quad \text{and} \quad W_{\Hm,\yv} \eqdef \left\{\ket{\phi} : \forall \ket{\psi} \in V_{\Hm,\yv}^{\neq}, \ \braket{\phi}{\psi} = 0\right\}.$$
\end{definition}
Notice that from Proposition~\ref{Proposition:FullDual}, if $S$ has full dual support then for each $k \in \Iint{0}{n}$ and $\forall (\Hm,\yv) \in \I_k$, we have $dim(V^{\neq}_{\Hm,\yv}) = 2^n(1 - \frac{1}{2^k})$ and hence $dim(W_{\Hm,\yv}) = 2^n - dim(V^{\neq}_{\Hm,\yv}) = 2^{n-k}$.
\begin{proposition}
	Let $S = \SS$ be a set of symmetric states with full dual support and let $\FF \in \Gamma(S)$. For each $k \in \Iint{0}{n}$ and $(\Hm,\yv) \in \I_k$, we can write 
	$$  F_{\Hm,\yv} = \sum_{i = 1}^{2^{n-k}} \mu_i \kb{B_i},$$
	where $\mu_i \ge 0$ and $\ket{B_i} \in  W_{\Hm,\yv}$ are pairwise orthogonal.	
\end{proposition}
\begin{proof}
	Fix $S = \SS$ a set of symmetric states with full dual support and $\FF \in \Gamma(S)$. Fix also $k \in \Iint{0}{n}$ and $(\Hm,\yv) \in \I_k$. Since, $\FF \in \Gamma(S)$, we have 
	\begin{align}\label{Eq:UM} \forall \xv \in \F_2^n, \ s.t. \ \Hm \xv \neq \yv, \ \Tr(F_{\Hm,\yv} \kb{\psi_{\xv}}) = 0.
	\end{align}
	Now because $F_{\Hm,\yv} \succeq \zerov$, we write the spectral decomposition 
	$$F_{\Hm,\yv} = \sum_i \mu_i \kb{B_i},$$
	where $\mu_i \ge 0$ and the $\ket{B_i}$ are pairwise orthogonal. From Equation~\ref{Eq:UM}, we have that 
	$$ \forall i,  \forall \xv \in \F_2^n, \ s.t. \ \Hm \xv \neq \yv, \ \braket{\psi_{\xv}}{B_i} = 0,$$
	which implies that $\ket{B_i} \in W_{\Hm,\yv}$. Since $dim(W_{\Hm,\yv}) = 2^{n-k}$,  this concludes the proof. 
\end{proof}

The next proposition provides an orthonormal basis for each $W_{\Hm,\yv}$. This construction will rely on the notion of dual cosets, which was presented in Definition~\ref{Definition:DualCosets}.
\begin{proposition}\label{Proposition:12}
	Let $S = \SS$ be a set of symmetric states with full dual support and let $\ket{\psi_\zerov} = \sum_{\iv \in \F_2^n} \halpha_\iv \ket{\hi}$. Let $\FF \in \Gamma(S)$. Fix $k \in \Iint{0}{n}$ and $(\Hm,\yv) \in \I_k$. For each $\sv \in \F_2^{n-k}$, let $\vv_{\sv} \in \F_2^n$ such that we can write 
	$$ D_{\Hm}(\sv) = \{\trp{\Hm}\uv + \vv_\sv : \uv \in \F_2^k\},$$
	(see Definition~\ref{Definition:DualCosets}).  We now define 
	\begin{align}\label{Eq:States} \forall \sv \in \F_2^{n-k}, \ \ket{A_\sv^{\Hm,\yv}} \eqdef \sum_{\uv \in \F_2^k} \frac{1}{\overline{\halpha_{\trp{\Hm}\uv + \vv_\sv}}} (-1)^{\yv \cdot \uv} \ket{\widehat{\trp{\Hm}\uv + \vv_\sv}}.
	\end{align}
	Then $\{\ket{A_\sv^{\Hm,\yv}}\}_{\sv \in \F_2^{n-k}}$ forms an orthogonal basis of $W_{\Hm,\yv}$.
\end{proposition}
\begin{proof}
	Fix $k \in \Iint{0}{n}$ and $(\Hm,\yv) \in \I_k$. First notice that each $\ket{A_{\sv}^{\Hm,\yv}}$ has support in $D_\Hm(\sv)$ hence
 $\braket{A_\sv^{\Hm,\yv}}{A_{\sv'}^{\Hm,\yv}} = 0$ when $\sv \neq \sv'$. Now, we prove that each $\ket{A_\sv^{\Hm,\yv}} \in W_{\Hm,\yv}$. It is enough to prove that
	$$ \forall \xv \in \F_2^n \ s.t. \ \Hm\xv \neq \yv,  \ \forall \sv \in \F_2^{n-k}, \ \braket{A_\sv^{\Hm,\yv}}{\psi_\xv} = 0.$$  
	So fix an $\xv \in \F_2^n$ such that $\Hm\xv \neq \yv$ as well as $\sv \in \F_2^{n-k}$. We have $\ket{\psi_\xv} = \sum_{\iv \in \F_2^n} \halpha_\iv (-1)^{\iv \cdot \xv}\ket{\hi}$. We first rewrite
	\begin{align*}
		\ket{\psi_{\xv}} & = \sum_{\sv \in \F_2^{n-k}} \sum_{\iv \in \D_\Hm(\sv)} \halpha_{\iv} (-1)^{\xv \cdot \iv}\ket{\hi} \\
		& = \sum_{\sv \in \F_2^{n-k}} \sum_{\uv \in \F_2^k} \halpha_{(\trp{\Hm}\uv + \vv_\sv)} (-1)^{\xv \cdot (\trp{\Hm}\uv + \vv_\sv)}\ket{\widehat{\trp{\Hm}\uv + \vv_\sv}} \\
		& = \sum_{\sv \in \F_2^{n-k}} (-1)^{\xv \cdot \vv_\sv} \sum_{\uv \in \F_2^k} \halpha_{(\trp{\Hm}\uv + \vv_\sv)} (-1)^{\xv \cdot \trp{\Hm}\uv}\ket{\widehat{\trp{\Hm}\uv + \vv_\sv}} \\
		& = \sum_{\sv \in \F_2^{n-k}} (-1)^{\xv \cdot \vv_\sv} \sum_{\uv \in \F_2^k} \halpha_{(\trp{\Hm}\uv + \vv_\sv)} (-1)^{\Hm \xv \cdot \uv}\ket{\widehat{\trp{\Hm}\uv + \vv_\sv}}
	\end{align*}
	From there, we can conclude 
	\begin{align*}
		\braket{{A_{\sv}^{\Hm,\yv}}}{{\psi_{\xv}}} & = (-1)^{\xv \cdot \vv_\sv} \sum_{\uv \in \F_2^k}  (-1)^{(\Hm \xv + \yv) \cdot \uv} = 0, \qquad \text{since } \Hm\xv \neq \yv
	\end{align*} 
We can conclude that the set $\{\ket{A_{\sv}^{\Hm,\yv}}\}_{\sv \in \F_2^{n-k}}$ is a set of pairwise orthogonal states with each state in $W_{\Hm,\yv}$. Since $dim(W_{\Hm,\yv}) = 2^{n-k}$, we conclude that $\{\ket{A_{\sv}^{\Hm,\yv}}\}_{\sv \in \F_2^{n-k}}$ is an orthogonal basis of $W_{\Hm,\yv}$. 
\end{proof}

\begin{theorem}\label{Theorem:LambdaRelations}
	Let $\FHy \in \Gamma(S)$. For each $k \in \Iint{0}{n}$, for each $(\Hm,\yv) \in \I_k$, for each $\sv \in \F_2^{n-k}$  we have 
	$$ \forall \iv,\jv \in \D_\Hm(\sv), \ |\halpha_\iv|^2 \triple{\hi}{F_{\Hm,\yv}}{\hi} = 
	|\halpha_\jv|^2 \triple{\hj}{F_{\Hm,\yv}}{\hj}.$$
\end{theorem}
\begin{proof}
	Fix any $k \in \Iint{0}{n}$ and $(\Hm,\yv) \in \I_k$. First notice that for each $\sv \in \F_2^{n-k}$, we have 
	$$ |\braket{\hi}{A_{\sv}^{\Hm,\yv}}|^2 = \left\{\begin{tabular}{cl}
		$ |\frac{1}{\halpha_{\iv}}|^2 $& { if } $\iv \in D_\Hm(\sv)$ \\
		$0$ & { if } $\iv \notin D_\Hm(\sv)$
	\end{tabular}\right.$$
	Now fix any $\ket{\phi} \in W_{\Hm,\yv}$. From Proposition~\ref{Proposition:12}, we have 
	$\ket{\phi} = \sum_{\sv \in \F_2^{n-k}} \gamma_\sv \ket{A_{\sv}^{\Hm,\yv}}$. Forall $\iv \in \D_\Hm(\sv)$.  We have
	\begin{align*}
		\braket{\hi}{\phi}\braket{\phi}{\hi} & = |\gamma_\sv|^2 \braket{\hi}{A_\sv^{\Hm,\yv}}\braket{A_\sv^{\Hm,\yv}}{\hi} = \frac{|\gamma_{\sv}|^2}{|\halpha_\iv|^2} 
	\end{align*}
	which implies that
	$$ \forall \iv,\jv \in \D_\Hm(\sv), \ \braket{\hi}{\phi}\braket{\phi}{\hi} |\halpha_\iv|^2 = \braket{\hj}{\phi}\braket{\phi}{\hj} |\halpha_\jv|^2.$$
 We can write $F_{\Hm,\yv} = \sum_l \mu_l \kb{B_l}$ where $\ket{B_l} \in W_{\Hm,\yv}$. From the above equality, we immediately obtain that 
	$$ \forall \iv,\jv \in \D_{\Hm}(\sv), \ |\halpha_\iv|^2 \triple{\hi}{F_{\Hm,\yv}}{\hi} = |\halpha_\jv|^2 \triple{\hj}{F_{\Hm,\yv}}{\hj}. \qedhere$$
\end{proof}

\subsection{Formulation of the linear program}
In Section~\ref{Sec:3.1}, we showed that the quantity $\rho(S)$ can be rewritten by optimizing over symmetric fine-grained measurements $\FHy$ but looking only at the Fourier diagonal terms of the matrices $F_{\Hm,\yv}$. In Section~\ref{Sec:3.2}, we showed that the unambiguity condition of fine-grained measurements implies a relation on the Fourier diagonal coefficients of $F_{\Hm,\yv}$. We now rewrite the optimization program $\rho(S)$ from Section~\ref{Sec:3.1} by replacing the unambiguity condition with the relation on these diagonal terms. This gives us a linear program which we now present

Let $S = \{\ket{\psi_{\xv}}\}_{\xv \in \F_2^n}$ be a set of symmetric states with $\ket{\psi_\zerov} = \sum_{\iv} \halpha_\iv \ket{\hi}$. \\

\vspace{-0.5em}

\begin{mdframed}[linewidth=1pt, roundcorner=5pt, nobreak=true, innerleftmargin=10pt, innerrightmargin=10pt, innertopmargin=8pt, innerbottommargin=8pt]
	\begin{center}
		\textbf{Final Linear Program}
	\end{center}
	
	\vspace{0.5em}
	
	\begin{align*}
		\text{Variables:} \quad & \lambda^\Hm_\iv \in \mathbb{R}_+ \quad \text{for each } \Hm \in \Lambda, \ \iv \in \mathbb{F}_2^n. \\
		\text{Objective:} \quad & \rho^L(S) \eqdef \max_{\lambda_\iv^\Hm} \sum_{k = 0}^n \sum_{\Hm \in \Lambda_k} \sum_{\iv \in \mathbb{F}_2^n} C(k) \, \lambda^\Hm_\iv \, |\halpha_\iv|^2. \\
		\text{Constraints:} \quad 
		& \text{\ding{172}} \quad \sum_{\Hm \in \Lambda} \lambda^\Hm_\iv = 1 
		&& \forall \iv \in \mathbb{F}_2^n, \\
		& \text{\ding{173}} \quad \lambda^\Hm_\iv \, |\halpha_\iv|^2 = \lambda^\Hm_\jv \, |\halpha_\jv|^2 
		&& \forall k \in \Iint{0}{n}, \Hm \in \Lambda_k, \sv \in \F_2^{n-k}, \ \forall \iv, \jv \in \D_\Hm(\sv)
	\end{align*}
\end{mdframed}


We constructed this linear program by taking properties of fine-grained unambiguous measurement so we have $\rho(S) \le \rho^L(S)$, which is captured by the proposition below. Notice that we use a more compact linear program, where we consider the variables $\lambda_\iv^\Hm = \sum_{\yv} \lambda_\iv^{\Hm,\yv}$. Proposition~\ref{Proposition:LinearReverse} shows this actually gives an equivalent linear program. 
\begin{proposition}
	For any set of symmetric states $S$ with full dual support, $\rho(S) \le \rho^L(S)$
\end{proposition}

\begin{proof}
	Let $S = \{\ket{\psi_\xv}\}_{\xv \in \F_2^n}$ be a set of states.
	Let $\{F_{\Hm,\yv}\} \in \Gamma_s(S)$ that maximizes the semidefinite program, meaning in particular that
	\begin{enumerate}
		\item $\forall k \in \Iint{0}{n}, \ \forall (\Hm,\yv) \in \Lambda_k \times \F_2^k, \ F_{\Hm,\yv} \succeq \zerov$,
		\item $\sum_{k = 0}^n \sum_{\substack{\Hm \in \Lambda_k \\ \yv \in \F_2^k}} F_{\Hm,\yv} = I$,
	\end{enumerate}
	and
	$$\rho(S) = \rho_2(S,\{F_{\Hm,\yv}\}) = \sum_{k \in \Iint{0}{n}}\sum_{(\Hm,\yv) \in \I_k}  \sum_{\iv \in \F_2^n} C(k)|\halpha_\iv|^2 \triple{\hi}{F_{\Hm,\yv}}{\hi},$$
	Our goal is to define $\lambda^{\Hm,\yv}_{\iv}$ and $\lambda^\Hm_\iv = \sum_{\yv} \lambda^{\Hm,\yv}_{\iv}$ that satisfy conditions \ding{172}, \ding{173} of the linear program and such that $\rho^L(S) = \rho_2(S,\{F_{\Hm,\yv}\})$.
	Choose $\lambda_\iv^{\Hm,\yv} = \bra{\hi} F_{\Hm,\yv} \ket{\hi}$.
	We derive from 1. that $\lambda_\iv^{\Hm,\yv} \ge 0$ and from 2. that $\sum_{(\Hm,\yv)\in\I} \lambda_\iv^{\Hm,\yv} = 1$.
	From Theorem~\ref{Theorem:LambdaRelations}, the unambiguity conditions translate into $\lambda_\iv^{\Hm,\yv}|\halpha_\iv|^2 = \lambda_\jv^{\Hm,\yv}|\halpha_\jv|^2$ for all $k\in\Iint{0}{n}, \ \Hm\in\Lambda_k, \ \sv\in\F_2^{n-k}$, for all $\iv,\jv\in\D_\Hm(\sv)$.
	Finally the objective becomes 
	$$\rho_2(S,\{F_{\Hm,\yv}\}) = \sum_{k \in \Iint{0}{n}}\sum_{(\Hm,\yv) \in \I_k}  \sum_{\iv \in \F_2^n} C(k)|\halpha_\iv|^2 \lambda_\iv^{\Hm,\yv} = \rho^L(S).$$
	We conclude that $\rho(S,\{F_{\Hm,\yv}\}) \le \rho^L(S)$ using Theorem~\ref{Theorem:4.1}.
\end{proof}

Actually, we can prove that the optimization programs are equal. This is an important property that tell us that there exist fine-grained unambiguous measurements that achieve the value given by the linear program.
\begin{proposition}\label{Proposition:LinearReverse}
	For any set of symmetric states $S$ with full dual support, $\rho^L(S) \le \rho(S)$.
\end{proposition}
\begin{proof}
	Let values $\lambda^\Hm_\iv \in \mathbb{R}_+$ for each $\Hm \in \Lambda$ and $\iv \in \F_2^n$ that maximize the linear program, meaning that they satisfy conditions \ding{172}, \ding{173} of the linear program and  
	$$ \rho^L(S) = \sum_{k = 0}^n \sum_{\Hm \in \Lambda_k} \sum_{\iv \in \mathbb{F}_2^n} C(k) \, \lambda^\Hm_\iv \, |\halpha_\iv|^2.$$
	Our goal is to construct a measurement $\FHy \in \Gamma_s(S)$ such that $\rho_2(S,\FHy) = \rho^L(S)$, which will imply the desired statement. Each $F_{\Hm,\yv}$ will be a one dimensional operator, \ie a scalar times a projector. For each $k \in \Iint{0}{n}$, for each $(\Hm,\yv) \in \I_k$, we fix 
	$$ F_{\Hm,\yv} = \kb{\Phi_{\Hm,\yv}} \quad \text{with} \quad \ket{\Phi_{\Hm,\yv}} = \sum_{\sv \in \F_2^{n-k}} \beta_\sv^{\Hm,\yv} \ket{A^{\Hm,\yv}_\sv},$$
	and the goal is to choose proper values of $\beta_{\sv}^{\Hm,\yv}$. First, notice that for any choice of $\beta_\sv^{\Hm,\yv} \in \mathbb{C}$, we have $F_{\Hm,\yv} \succeq \zerov$. Moreover, $\ket{\Phi_{\Hm,\yv}} \in W_{\Hm,\yv}$ hence it will satisfy the fine-grained unambiguous condition. For the amplitudes, we choose them such that 
	\begin{align}
		\beta_\sv^{\Hm,\yv} \eqdef \sqrt{\lambda^{\Hm}_\iv |\halpha_\iv|^2} \quad \text{for any } \iv \in \D_\Hm(\sv).
	\end{align}
	Notice that this amplitude is independent of $\yv$. 
	Now, 
	\begin{align*}
		\forall \iv \in \D_\Hm(\sv),  \ \triple{\hi}{F_{\Hm,\yv}}{\hi} = |\braket{\hi}{\Phi_{\Hm,\yv}}|^2 = |\beta_\sv^{\Hm,\yv}|^2 \cdot |\braket{\hi}{A_\sv^{\Hm,\yv}}|^2 = \frac{|\beta_\sv^{\Hm,\yv}|^2}{|\halpha_\iv|^2} = \lambda_\iv^{\Hm},
	\end{align*}
	where we used Equation~\ref{Eq:States}. We then obtain 
	\begin{align}\label{Eq:Partial} \rho^L(S) = \sum_{k = 0}^n \sum_{\Hm \in \Lambda_k} \sum_{\iv \in \mathbb{F}_2^n} C(k) \, \triple{\hi}{F_{\Hm,\yv}}{\hi}  \, |\halpha_\iv|^2.\end{align}
	We can also write 
	$$\sum_{(\Hm,\yv) \in \I^*} F_{\Hm,\yv} = \sum_{\Hm \in \Lambda^*} \sum_{\iv \in \F_2^n} \lambda_\iv^\Hm \kb{\hi} \preceq I,$$
	using the condition on the $\lambda_\iv^\Hm$. We proved that $\FHy \in \Gamma(S)$ but we actually want to prove that $\FHy \in \Gamma_s(S)$. We write 
	$$
	F_{\Hm,\yv} = \sum_{\sv \in \F_2^{n-k}} |\beta_\sv^\Hm|^2 \kb{A_\sv^{\Hm,\yv}}.$$
	Now fix any $\av \in \F_2^n$. We have from Equation~\ref{Eq:States}
	\begin{align*}
		X_a \ket{A_\sv^{\Hm,\yv}} & = \sum_{\uv \in \F_2^k} \frac{1}{\halpha_{\trp{\Hm}\uv + \vv_\sv}} (-1)^{\yv \cdot \uv} (-1)^{\av \cdot \trp{\Hm}\uv + \vv_\sv} \ket{\widehat{\trp{\Hm}\uv + \vv_\sv}} \\
		& = (-1)^{\av \cdot \vv_\sv} \sum_{\uv \in \F_2^k} \frac{1}{\halpha_{\trp{\Hm}\uv + \vv_\sv}} (-1)^{\yv \cdot \uv} (-1)^{\Hm\av \cdot \uv} \ket{\widehat{\trp{\Hm}\uv + \vv_\sv}} \\
		& = (-1)^{\av \cdot \vv_\sv} \ket{A_\sv^{\Hm,\yv + \Hm \av}}.
	\end{align*}
	From there, we obtain 
	\begin{align*}
		X_a \sum_{\sv \in \F_2^{n-k}} \beta_{\sv}^{\Hm,\yv} \ket{A_{\sv}^{\Hm,\yv}} 
		& = \sum_{\sv \in \F_2^{n-k}} (-1)^{\av \cdot v_{\sv}} \beta_{\sv}^{\Hm,\yv}\ket{A_\sv^{\Hm,\yv + \Hm\av}} \\
		& = \sum_{\sv \in \F_2^{n-k}} (-1)^{\av \cdot v_\sv} \beta_{\sv}^{\Hm,\yv + \Hm \av}\ket{A_\sv^{\Hm,\yv + \Hm\av}},
	\end{align*}
	where the last inequality comes from the fact that the $\beta_{\sv}^{\Hm,\yv}$ are actually independent of $\yv$. From there we conclude that $X_\av F_{\Hm,\yv} X_{\av} = F_{\Hm,\yv + \Hm\av}$ and hence $\FHy \in \Gamma_s(S)$.
	
	In conclusion, from the above and Equation~\ref{Eq:Partial}, we constructed $\FHy \in \Gamma_s(S)$ such that 
	$$\rho^L(S) = \sum_{k = 0}^n \sum_{\Hm \in \Lambda_k} \sum_{\iv \in \mathbb{F}_2^n} C(k) \, \triple{\hi}{F_{\Hm,\yv}}{\hi}  \, |\halpha_\iv|^2 = \rho_2(S,\FHy),$$ which allows us to conclude that $\rho^L(S) \le \rho(S)$ using Theorem~\ref{Theorem:4.1}.
\end{proof}

\subsection{Removing the full support requirement}
In the case of sets $S$ of symmetric states with full support, we defined the value $\rho^L(S)$ and showed that $\rho^L(S) = \rho(S)$. Also, both $\rho^L(S)$ and $\rho(S)$ are well defined when $S$ does not have full support. We will be able to remove this requirement by a simple density argument. We have the following
\begin{proposition}
	Let $S$ be a set of symmetric states and let $\eps > 0$. There exists a set of symmetric states $S'$ with full dual support such that 
	$$ |\rho(S') - \rho(S)| \le \eps  \quad \text{and} \quad |\rho^L(S') - \rho^L(S)| \le \eps.$$
\end{proposition}
\begin{proof}
	Let $\eps > 0$ and let $S = \{\ket{\psi_\xv}\}$ be a set of symmetric states with $\ket{\psi_\zerov} = \sum_{\iv \in \F_2^n} \halpha_{\iv} \ket{\hi}$. Let $T = \{\iv \in \F_2^n : \halpha_\iv = 0\}$. If $T = \emptyset$, $S$ has full support and the statement is trivial. We now look at the case where $T \neq \emptyset$. Let $\delta = \min\{\frac{\eps }{\rho^L(S) + 1} ; \frac{\eps}{\rho(S) + 1} ; \frac{\eps}{C_{Max}} ; 1\} > 0$, where $C_{Max} = \max_{k \in \Iint{0}{n}} C(k) > 0$. We define 
	$$ \ket{\psi'_\xv} = \sqrt{1-\delta} \ket{\psi_\xv} + \sqrt{\frac{\delta}{|T|}} \sum_{\iv \in T} (-1)^{\iv\cdot\xv} \ket{\hi}.$$
	The set $S' = \{\ket{\psi'_\xv}\}$ clearly is symmetric and has full dual support and we write $\ket{\psi'_\zerov} = \sum_{\iv \in \F_2^n} \halpha'_\iv \ket{\hi}$. This means that $\halpha'_\iv = \sqrt{1-\delta} \halpha_\iv$ for $\iv \notin T$ and $\halpha'_\iv = (-1)^{\iv\cdot\xv}\sqrt{\frac{\delta}{|T|}}$ for $\iv \in T$. 
	
	In order to prove the first inequality, let $\FHy \in \Gamma_s(S')$ that maximizes $\rho(S')$ in the expression of Theorem~\ref{Theorem:4.1}. We write
	\begin{align*}
		\rho(S') & = \sum_{k \in \Iint{0}{n}}\sum_{(\Hm,\yv) \in \I_k} \sum_{\iv \in \F_2^n} C(k)|\halpha'_\iv|^2 \triple{\hi}{F_{\Hm,\yv}}{\hi} \\
		& = \sum_{k \in \Iint{0}{n}}\sum_{(\Hm,\yv) \in \I_k}\sum_{\iv \notin T} C(k)|\halpha'_\iv|^2 \triple{\hi}{F_{\Hm,\yv}}{\hi} + \sum_{k \in \Iint{0}{n}}\sum_{(\Hm,\yv) \in \I_k}\sum_{\iv \in T} C(k)|\halpha'_\iv|^2 \triple{\hi}{F_{\Hm,\yv}}{\hi} \\
		& \le (1-\delta)\rho(S) + \frac{\delta}{|T|} C_{Max}    
	\end{align*} 
where we used 	
$$\sum_{k \in \Iint{0}{n}}\sum_{(\Hm,\yv) \in \I_k}\sum_{\iv \in T} \triple{\hi}{F_{\Hm,\yv}}{\hi} \le \sum_{k \in \Iint{0}{n}}\sum_{(\Hm,\yv) \in \I_k}\sum_{\iv \in \F_2^n} \triple{\hi}{F_{\Hm,\yv}}{\hi} \le 1.$$	
From there, we obtain
$$ \rho(S') - \rho(S) \le -\delta\rho(S) + \delta C_{Max} \le \delta C_{Max} \le \eps,$$
which gives the first part inequality. For the second direction of the first inequality, let $\FHy \in \Gamma_s(S)$ that maximizes $\rho(S)$. We write 
\begin{align*}
	\rho(S') & \ge \sum_{k \in \Iint{0}{n}}\sum_{(\Hm,\yv) \in \I_k} \sum_{\iv \in \F_2^n} C(k)|\halpha'_\iv|^2 \triple{\hi}{F_{\Hm,\yv}}{\hi} \\
	& \ge (1-\delta) \sum_{k \in \Iint{0}{n}}\sum_{(\Hm,\yv) \in \I_k} \sum_{\iv \in \F_2^n} C(k)|\halpha_\iv|^2 \triple{\hi}{F_{\Hm,\yv}}{\hi} = (1-\delta)\rho(S),
\end{align*}
which directly implies 
$$\rho(S') - \rho(S) \ge - \delta \rho(S) \ge \frac{- \eps \rho(S)}{\rho(S) + 1} \ge -\eps.$$ 

\noindent For the second inequality, let $\{\lambda_\iv^\Hm\}$ that maximizes the linear program for $\rho^L(S)$. We now write
	\begin{align*}\rho^L(S') & \ge \sum_{k = 0}^n \sum_{\Hm \in \Lambda_k} \sum_{\iv \in \mathbb{F}_2^n} C(k) \, \lambda^\Hm_\iv \, |\halpha'_\iv|^2 \\
	& \ge (1-\delta) \max \sum_{k = 0}^n \sum_{\Hm \in \Lambda_k} \sum_{\iv \in \mathbb{F}_2^n} C(k) \, \lambda^\Hm_\iv \, |\halpha_\iv|^2 = (1-\delta) \rho^L(S) \end{align*}
	which gives 
	$$ \rho^L(S') - \rho^L(S) \ge - \delta \rho^L(S) \ge -\eps \frac{\rho^L(S)}{\rho^L(S) + 1} \ge -\eps.$$
For the second part of the second inequality, consider $\{\lambda_\iv^\Hm\}$ that maximizes the linear program for $\rho^L(S')$. We write
	\begin{align*}
		\rho^L(S') & = \sum_{k = 0}^n \sum_{\Hm \in \Lambda_k} \sum_{\iv \in \mathbb{F}_2^n} C(k) \, \lambda^\Hm_\iv \, |\halpha_\iv|^2  \\
		& \le (1-\delta) \rho^L(S) + \sum_{k = 0}^n\sum_{\Hm \in \Lambda_k}\sum_{\iv \in T} C(k)\lambda^\Hm_\iv |\halpha'_\iv|^2 \\
		& \le (1-\delta) \rho^L(S) + \frac{\delta}{|T|} C_{Max}
	\end{align*}
	which gives 
	$$ \rho^L(S') - \rho^L(S) \le -\delta \rho^L(S) + \frac{\delta}{|T|} C_{Max} \le \delta C_{Max} \le \eps. \qedhere$$
\end{proof}

From there, we immediately obtain
\begin{theorem}
	Let $S$ be a set of symmetric states. Then $\rho(S) = \rho^L(S)$.
\end{theorem}
\begin{proof}
	Assume by contradiction that $\rho^L(S) \neq \rho(S)$. Let $\eps = \frac{|\rho^L(S) - \rho(S)|}{3} > 0$. From the previous proposition, let $S'$ be a set of states of dual support such that 
	$$ |\rho(S') - \rho(S)| \le \eps  \quad \text{and} \quad |\rho^L(S') - \rho^L(S)| \le \eps.$$
	Now, because $S'$ has full dual support, we have $\rho(S') = \rho^L(S')$. We can therefore conclude
	$$|\rho^L(S) - \rho(S)| \le |\rho^L(S) - \rho^L(S')| + |\rho^L(S') - \rho(S')| + |\rho(S') - \rho(S)| \le 2\eps,$$
	which contradicts the fact that $|\rho^L(S) - \rho(S)| = 3\eps$ with $\eps > 0$.
\end{proof}

\section{Dual Linear Program and Solutions}
Let $S = \{\ket{\psi_\xv}\}_{\xv \in \F_2^n}$ be a set of symmetric states. We write $\ket{\psi_\zerov} = \sum_{\iv \in \F_2^n} \halpha_\iv \ket{\hi}$. Recall our goal is to give bounds on $\rho(S)$. In the previous section, we showed that $\rho(S) = \rho^L(S)$ where $\rho^L(S)$ can be expressed as a linear maximization problem. In this section, we present the associated dual linear program, and derive upper bounds on the value of $\rho(S)$, both in the threshold regime and in the average parity regime. 
\subsection{Formulation of the dual program}
From the linear program with objective function $\rho^L(S)$, we construct the associated dual linear program, with objective $\sigma^L(S)$.

\begin{mdframed}[linewidth=1pt, roundcorner=5pt, nobreak=true, innerleftmargin=10pt, innerrightmargin=10pt, innertopmargin=8pt, innerbottommargin=8pt]
	\begin{center}
		\textbf{Dual Linear Program}
	\end{center}
	
	\vspace{0.5em}
	
	\begin{align*}
		\text{Variables:} \quad & b_\iv \in \mathbb{R}_+ \quad \text{for each } \iv \in \mathbb{F}_2^n. \\
		\text{Objective:} \quad & \sigma^L(S) \eqdef \min_{b_\iv} \sum_{\iv\in\F_2^n} b_\iv |\hat{\alpha}_\iv|^2 \\
		\text{Constraints:} \quad 
		& \sum_{\iv\in \D_\Hm(\sv)} b_{\iv} \ge C(k) 2^k 
		&& \forall k\in\Iint{0}{n}, \ \forall \Hm \in \Lambda_k, \ \forall \sv \in \F_2^{n-k}.
	\end{align*}
\end{mdframed}

Remark that we should normally take $b_\iv\in\mathbb{R}$ but using the constraints for $k=0$, $\Hm=\begin{pmatrix}0 & \dots & 0\end{pmatrix}$ and the fact that $C:\Iint{0}{n} \mapsto \mathbb{R}_{+}$, we get that $\forall \sv\in\F_2^n, \ b_\sv = \sum_{\iv\in\D_{\Hm}(\sv)} b_\iv \ge C(0) \ge 0$.

\COMMENT{
\begin{proposition}[Weak duality]
	$$\rho^L(S) \le \sigma^L(S).$$
\end{proposition}

\begin{proof}
	\begin{align*}
		\sigma^L(S) & =  \sum_{\iv\in \F_2^n} b_{\iv} |\hat{\alpha}_\iv|^2 = \sigma^L(S) \\
		& = \sum_{\iv\in \F_2^n} b_{\iv} |\hat{\alpha}_\iv|^2 \sum_{\Hm \in \Lambda^*} \lambda_\iv^\Hm & \text{from primal constraint } 1. \\
		& = \sum_{k = 1}^n \sum_{\Hm \in \Lambda_k} \sum_{\sv\in \F_2^{n-k}} \sum_{\iv\in \D_\Hm(\sv)} b_{\iv} |\hat{\alpha}_\iv|^2 \lambda_\iv^\Hm
	\end{align*}
	Now, from primal constraint $2$, for each $\Hm \in \Lambda_k, \ \sv \in \F_2^k$ and $\iv \in \D_\Hm(\sv)$, we have $|\hat{\alpha}_\iv|^2 \lambda_\iv^\Hm = \frac{1}{2^k} \sum_{\jv \in \D_\Hm(\sv)} |\hat{\alpha}_\jv|^2 \lambda_\jv^\Hm$. Pluggin this in the above, we obtain
	\begin{align*}
		\sigma^L(S) & = \sum_{k = 1}^n \sum_{\Hm \in \Lambda_k} \sum_{\sv\in \F_2^{n-k}} \frac{1}{2^k} \sum_{\jv \in \D_\Hm(\sv)}  |\hat{\alpha}_\jv|^2 \lambda_\jv^\Hm \sum_{\iv\in \D_\Hm(\sv)} b_{\iv} \\
		& \ge \sum_{k = 1}^n \sum_{\Hm \in \Lambda_k} \sum_{\sv\in \F_2^{n-k}} \sum_{\jv \in \D_\Hm(\sv)}  |\hat{\alpha}_\jv|^2 \lambda_\jv^\Hm C(k) & \text{from dual constraint } 2. \\
		& = \rho^L(S)
	\end{align*}
\end{proof}
}
\begin{proposition}[Strong duality]
	$$\rho^L(S) = \sigma^L(S).$$	
\end{proposition}

\begin{proof}
	$(\lambda_\iv^\Hm)_{\Hm\in\Lambda, \iv\in\F_2^n} = 0$ is a primal solution whatever the $|\halpha_\iv^2|$.
	$(b_\iv) = 2^n \max_{k\in\Iint{0}{n}} C(k)$ is a dual solution whatever the $|\halpha_\iv^2|$.
	Thus, by strong duality, the primal has an optimal solution $(\lambda_\iv^{\Hm,*})$, the dual has an optimal solution $(b_\iv^{*})$ and $\rho^L((\lambda_\iv^{\Hm,*}), S) = \sigma^L((b_\iv^{*}), S)$.
\end{proof}

\subsection{The threshold setting}
Let $\tau \in \Iint{1}{n}$. We first consider the threshold setting meaning that we fix $C(k) = 1$ if $k \ge \tau$ and $C(k) = 0$ otherwise. This means we want to bound the probability that we can learn unambiguously at least $\tau$ parities of $\xv$ given $\ket{\psi_\xv}$. We write our objective function $\rho(S;\tau)$ as well as the associated primal and dual objective functions respectively  $\rho^L(S;\tau)$ and $\sigma^L(S;\tau)$. From our previous results, we have $\rho(S;\tau) = \rho^L(S;\tau) = \sigma^L(S;\tau)$.
First, we give a necessary and sufficient condition for the existence of unambiguous measurements in this setting. The key concept here is the notion of $k$-universal set.

\begin{definition}
	Let $\aa_2(n,k)$ be the set of affine subspaces of $\F_2^n$ of dimension $k$. $U \in \aa_2(n,k)$ if it is of the form $\vv + V$ with $\vv \in \F_2^n$ and $V \in \G_2(n,k)$.
\end{definition}

\begin{definition}
	A set $U \subseteq \F_2^n$ is called $k$-universal iff. 
	$$ \forall V \in \aa_2(n,k),  U \cap V \neq \emptyset.$$
\end{definition}

This means that $U$ intersects every $k$-dimensional affine subspace of $\F_2^n$. An equivalent formulation is given below.

\begin{proposition}
		$U$ is $k$-universal iff.  
	$$ \forall \Hm \in \Lambda_{k}, \  \forall \sv \in \F_2^{n-k}, \  U \cap \D_\Hm(\sv) \neq \emptyset.$$
\end{proposition}
\begin{proof}
	This comes directly from the fact that an affine space $V \in \aa_2(n,k)$ can be written $V = \{\xv \in \F_2^n : \Gm \xv = \sv\}$ for some $\Gm \in \Lambda_{n-k}$ and $\sv \in \F_2^{n-k}$.
\end{proof}


We show the following
\begin{theorem}
	Let $S = \{\ket{\psi_\xv}\}_{\xv \in \F_2^n}$ be a set of symmetric states with $\ket{\psi_{\zerov}} = \sum_{\iv \in \F_2^n} \halpha_\iv \ket{\hi}$. 
		$$ \rho(S,\tau) = 0 \Leftrightarrow \text{There exists a } \tau\text{-universal set } V \ s.t. \ \forall \iv \in V, \ \halpha_\iv = 0.$$
\end{theorem}

We prove both directions of the equivalence. For the first direction, we actually prove a quantitative statement. 

\begin{proposition}
	$$ \rho(S;\tau) = \sigma^L(S;\tau) \le \min\left\{2^{\tau} \sum_{\iv \in V} |\halpha_\iv|^2 : V \text{ is } \tau\text{-universal} \right\}.$$
\end{proposition}
\begin{proof}
	Fix a threshold $\tau$ and an $\tau$-universal set $V$. As a dual solution, we fix $b_{\iv} = 2^\tau$ if $\iv \in V$ and $b_{\iv} = 0$ otherwise. We clearly have as objective $2^\tau \sum_{\iv \in V} |\halpha_\iv|^2$. We now have to check the constraints. We require that 
	\begin{align}\label{Eq:Requirement}\forall k \in \Iint{\tau}{n}, \ \forall \Hm \in \Lambda_k, \forall \sv \in \F_2^{n-k}, \sum_{{\iv \in \D_\Hm(\sv)}} b_\iv \ge 2^k.
	\end{align}
	For each $\Hm \in \Lambda_k$, we associate a matrix $\Gm_\Hm \in \Lambda_{n-k}$ such that $D_\Hm(\sv) = \{\iv \in \F_2^n : \Gm_\Hm \cdot \iv = \sv\}$.
	With our choice of $(b_\iv)$, the requirement of Equation~\ref{Eq:Requirement} is equivalent to 
	$$\forall k \in \Iint{\tau}{n}, \ \forall \Hm \in \Lambda_k, \forall \sv \in \F_2^{n-k}, |\left\{\iv \in V : \Gm_\Hm \cdot \iv = \sv \right\}| \ge 2^{k - \tau}.$$
	Fix $k \in \Iint{\tau}{n},\Hm,\sv$. We fix the associated matrix $\Gm_\Hm \in \F_2^{(n-k) \times n}$. We add lines to $\Gm_\Hm$ so that we have a matrix $\Mm \in \Lambda_{n-\tau}$. From the $\tau$-universality condition, we know that for any $\sv' \in \F_2^{k - \tau}$, there exists $\yv \in V$ such that $\Mm \yv = \sv || \sv'$. Since this condition implies that $\Gm_\Hm \cdot \yv = \sv$, we constructed $2^{k - \tau}$ different strings $\yv \in V$ such that $\Gm_\Hm \cdot \yv = \sv$ which concludes the proof.
\end{proof}
As a direct corollary, we have the following
\begin{corollary}
	If there exists an $\tau$-universal set $V$ such that $\forall \iv \in V, \ \halpha_\iv = 0$ then $\rho(S;\tau) = 0$.
\end{corollary}
We now prove the reverse implication. 
\begin{proposition}
	If $\rho(S;\tau) = \sigma^L(S;\tau) = 0$, then there exists an $\tau$-universal set $V$ such that $\forall \iv \in V, \halpha_\iv = 0$.
\end{proposition}
\begin{proof}
	We assume $\sigma^L(S) = 0$. Let $(b_i)_{\iv \in F_2^n}$ be an optimal dual solution meaning that $\sum_{\iv \in \F_2^n} b_i |\halpha_\iv|^2= 0$. Let $T = \{\iv : b_\iv \neq 0\}$. Since the $(b_\iv)_{\iv \in \F_2^n}$ satisfies the dual constraints, we have that 
	$$ \forall k \ge \tau, \ \forall \Hm \in \Lambda_k, \ \forall \sv \in \F_2^{n-k}, \ \sum_{\iv \in \D_\Hm(\sv)}b_\iv \ge 2^{k} > 0.$$
	This in particular implies that 
	$$  \forall k \ge \tau, \ \forall \Hm \in \Lambda_k, \ \forall \sv \in \F_2^{n-k}, \ T \cap \D_\Hm(\sv) \neq \emptyset,$$
	which implies that $T$ contains a $\tau$-universal set $V$.
	Now, because $\sigma^L(S) = \sum_{\iv \in T} b_{\iv}|\halpha_\iv|^2 = 0$, we have that $\forall \iv \in T, \ \halpha_\iv = 0$, which concludes the proof since 
	$V \subseteq T$. 
\end{proof}

Now, we want to show that if the $|\halpha_\iv|^2$ are concentrated on words of weight $< \frac{\tau}{2}$, then one can cannot learn unambiguously $\tau$ parities of $\xv$ given $\ket{\psi_{\xv}}$. Let $B_d \eqdef \{\xv \in \F_2^n : |\xv|_H \le d\}$. We first prove the following

\begin{proposition}\label{Proposition:Affine}
	For any $V \in \aa_2(n,k)$, $|V \cap B_d| \le \sum_{a=0}^d \binom{k}{d}$.
\end{proposition}
\begin{proof}
	Let $V \in \aa_2(n,k)$ and we write $V = \vv + W$ where $W$ is a linear subspace of dimension $k$. Let $\Gm \in \F_2^{k \times n}$ such that $W = \{\trp{\Gm} \sv : \sv \in \F_2^k\}$. Notice that 
	$$ V \cap B_d = \{\sv : |\trp{\Gm} \sv + \vv|_H \le d \}.$$
	Because $\Gm$ is of full rank, there exists $I \subseteq \Iint{1}{n}$ with $|I| = k$ such that $\Gm_I$ (the matrix where the columns of $\Gm$ are restricted to those with indices is $I$) is a square matrix of full rank $k$. $\trp{(\Gm_I)} \sv$ spans the whole space $\F_2^k$ for $\sv \in \F_2^k$ therefore 
	$$ \left|\{\sv : |\trp{(\Gm_I)} \sv + \vv_I|_H \le d\}\right| \le \sum_{a = 0}^d \binom{k}{a}.$$
	In order to conclude, notice that $|\trp{\Gm} \sv+ {\vv}|_H \ge |\trp{(\Gm_I)} \sv + \vv_I|_H$, from which we get
	$$
	|V \cap B_d| = \left|\{\sv : |\trp{\Gm} \sv + \vv|_H \le d\}\right| \le
	\left|\{\sv : |\trp{(\Gm_I)} \sv + \vv_I|_H \le d\}\right| = \sum_{a = 0}^d \binom{k}{a},$$
	which concludes the proof. 
\end{proof}

We can now prove our main threshold theorem.
\begin{theorem}
	Let $S = \{\ket{\psi_{\xv}}\}_{\xv \in \F_2^n}$ be a set of symmetric states with $\ket{\psi_\zerov} = \sum_{\iv \in \F_2^n} \halpha_{\iv} \ket{\hi}$. Let $\gamma > 2$ be an absolute constant. Let $\eps$ such that $\sum_{\iv \notin B_d} |\halpha_\iv|^2 = \eps$. Then $\rho(S,\lceil\gamma d\rceil) \le \eps(1 + o(1))$, where $o(1)$ is a quantity that goes to $0$ as $d,n \rightarrow \infty$.
\end{theorem}
\begin{proof}
	We prove this by constructing a solution to the dual program. We fix $\gamma > 2$ and let $\tau = \lceil \gamma d \rceil$. We choose as dual solution 
	\begin{align*}
		b_\iv  = \frac{2^\tau}{2^\tau - \sum_{a = 0}^d \binom{\tau}{a}}  \text{ for } \iv \notin B_d \quad ; \quad 
			b_\iv  = 0  \text{ for } \iv \in B_d
	\end{align*}
	We now check each constraint. Pick $k \in \Iint{\tau}{n}$, $\Hm \in \Lambda_k$ as well as $\sv \in \F_2^{n-k}$. The dual constraint can be written 
	$$ \sum_{\iv \in \D_\Hm(\sv)} b_\iv \ge 2^k,$$
	and we now prove it is satisfied. We write 
	\begin{align*}
		\sum_{\iv \in \D_\Hm(\sv)} b_\iv & = \frac{2^\tau}{2^{\tau} - \sum_{a = 0}^d \binom{\tau}{a}} |\D_\Hm(\sv) \cap \overline{B_d}| \\
		& \ge \frac{2^\tau}{2^{\tau} - \sum_{a = 0}^d \binom{\tau}{a}} \left(2^k - \sum_{a = 0}^d \binom{k}{a}\right) & \text{ from Proposition}~\ref{Proposition:Affine} \\
		& \ge 2^k \frac{2^\tau - \frac{1}{2^{k - \tau}}\sum_{a = 0}^{d}\binom{k}{a}}{2^\tau - \sum_{a = 0}^{d} \binom{\tau}{a}}
	\end{align*}
	In order to conclude, notice that $d \le \tau/2$ 
	which allows us to write for each $a \in \Iint{0}{d}$
		$$\frac{\binom{k}{a}}{\binom{\tau}{a}} = \prod_{i=\tau}^{k-1} \frac{\binom{i+1}{a}}{\binom{i}{a}} = \prod_{i=\tau}^{k-1} \frac{i+1}{i+1-a} \le \prod_{i=\tau}^{k-1} \frac{i+1}{i+1-i/2} \le 2^{k-\tau}$$
	plugging this in the above inequality, we obtain indeed that 
	$$\sum_{\iv \in \D_\Hm(\sv)} b_\iv \ge 2^k \frac{2^\tau - \sum_{a = 0}^{d}\binom{\tau}{a}}{2^\tau - \sum_{a = 0}^{d} \binom{\tau}{a}} = 2^k,$$
	which concludes the proof. 
\end{proof}

The way to interpret this proposition is the following. Assume the set of states $S$ has the property that a $(1 - negl(n))$ of weight of the $|\halpha_\iv|$ lies in strings $\iv$ of weight smaller than $\frac{\tau}{2}$, then one can learn $\tau$ parities of $\xv$ from $\ket{\psi_\xv}$ only with negligible probability.

\subsection{The average number of parities setting}
\subsubsection{Upper bounds} \label{sec:avg_bounds}
In the average setting, we want to bound the average number of parities that can be learned unambiguously. This means we fix $C(k) = k$. We write our primal and dual objective function respectively $\rho^L_{Av}(S)$ and $\sigma^L_{Av}(S)$. In this setting, we will prove two bounds. The first bound is related to the average dual weight of $\ket{\psi_\zerov}$.

\begin{theorem}\label{Theorem:Hamming_bound}
	$$\sigma^L_{Av}(S) \le 2 \sum_{\iv \in \F_2^n} |\iv|_H|\halpha_\iv|^2.$$
\end{theorem}
\begin{proof}
	We take a dual solution $b_\iv = 2|\iv|_H$. We now prove that this solution satisfies the dual constraints. We first prove a lemma on the minimum average weight of vectors in an affine subspace of dimension $k$.
	\begin{proposition}
		For any $\Hm \in \F_2^{k \times n}$, for any $\sv \in \F_2^{n-k}$, we have $\sum_{\iv \in \D_\Hm(\sv)} 2|\iv|_H \ge k 2^k.$
	\end{proposition}
	\begin{proof}
		We fix $\Hm \in \F_2^{k \times n}$, for any $\sv \in \F_2^{n-k}$. Let $\Gm_\Hm \in \F_2^{(n-k) \times n}$ be the matrix such that $\D_\Hm(\sv) = \{\xv \in \F_2^n : \Gm_\Hm \xv = \sv\}$. We fix a vector $\vv$ such that $\Gm_\Hm \vv = \sv$, which means we can also write 
		$$ \D_\Hm(\sv) = \{\trp{\Hm}\xv + \vv : \xv \in \F_2^k\}.$$
		Now, let $I \subseteq \Iint{1}{n}$ with $|I| = k$ such that $\Hm_I$ (where we restrict columns to the ones with indices in $I$) is a  square matrix of full rank $k$. We write 
		$$ \sum_{\iv \in \D_\Hm(\sv)} |\iv|_H = \sum_{\xv \in \F_2^k} |\trp{\Hm}\xv + \vv|_H \ge \sum_{\iv \in \D_\Hm(\sv)} |\trp{(\Hm_I)}\xv + \vv_I|_H = \frac{k 2^k}{2}.$$
		where in the last equality, we used the fact that the $\trp{(\Hm_I)}\xv$ hits each element of $\F_2^k$ exactly once, since $\Hm_I$ is a square matrix of full rank $k$. \end{proof}
		
		From the above lemma, we immediately have that 
		$$ \sum_{\iv \in \D_\Hm(\sv)} b_\iv = 2 \sum_{\iv \in \D_\Hm(\sv)} |\iv|_H \ge k 2^k.$$  
	This implies that the $b_\iv = 2|\iv|_H$ satisfies the dual constraints. We therefore conclude that 
	$$\rho_{Av}(S) = \sigma^L_{Av}(S) \le \sum_{\iv \in \F_2^n} 2|\iv|_H|\halpha_\iv|^2.$$
\end{proof}
As a direct corollary, we have that if the $|\halpha_\iv|^2$ are fully concentrated around words of weight at most $d$, meaning that 
$\sum_{\iv : |\iv|_H \le d} |\halpha_\iv|^2 \ge 1-\eps$, then one can learn on average at most $2\left(d(1-\eps) + n\eps\right)$ parities which is $\approx 2d$ when $\eps \approx 0$. This theorem is therefore the average case equivalent of Theorem~\ref{Proposition:ThresholdBound}.

The above upper bound is sometimes optimal but in many instances, it is also far from optimal. We provide a second family of upper bounds for the average case setting.

\begin{proposition}\label{Proposition:22}
	$$\sigma^L_{Av}(S) \le (2^n+n-1)|\halpha_\zerov|^2 + (n-1) \sum_{\iv \in \F_2^n \setminus \{\zerov\}} |\halpha_\iv|^2.$$
\end{proposition}

\begin{proof}
We take a dual solution such that $b_\zerov = 2^n+n-1$ and $b_\iv = n-1$ for $\iv\ne\zerov\in\F_2^n$. We have to prove that this solution satisfies the dual constraints. Fix any $k \in \Iint{1}{n}$, $\Hm \in \Lambda_k$ and $\sv \in \F_2^{n-k}$. We have to prove that 
$$\sum_{\iv \in \D_\Hm(\sv)} b_\iv \ge k 2^k.$$
	We distinguish two cases: 
	\begin{itemize}
		\item Either $\zerov\in \D_\Hm(\sv)$ thus,
		$\left( \sum_{\iv\in \D_\Hm(\sv)} b_{\iv} \right) = 2^n+n-1 + (|\D_\Hm(\sv)| - 1) (n-1) = 2^n+n-1 + (2^k - 1) (n-1) = n2^k + 2^n - 2^k \ge k2^k$.
		\item Or $\zerov \notin \D_\Hm(\sv)$ thus,
		$\left( \sum_{\iv\in \D_\Hm(\sv)} b_{\iv} \right) =  |\D_\Hm(\sv)| (n-1) = 2^k (n-1) \ge k2^k$,
		note that this case can only occur when $k<n$, since for $k=n$, we necessarily have $\zerov \in \D_\Hm(\sv)$.
	\end{itemize}
	Consequently, the dual constraints are saturated precisely when $k=n$ or $k=n-1$ and $\zerov \notin \D_\Hm(\sv)$.
This implies that the $(b_\iv)_{\iv\in\F_2^n}$ satisfies the dual constraints. We therefore conclude that 
$$\sigma^L_{Av}(S) \le (2^n+n-1)|\halpha_\zerov|^2 + (n-1) \sum_{\iv \in \F_2^n \setminus \{\zerov\}} |\halpha_\iv|^2. \qedhere$$
\end{proof}

From the following proposition, additional dual solutions and upper bounds can be easily obtained.

\begin{proposition}\label{Proposition:dual_perm}
	Let $\bv = (b_\iv)_{\iv \in \F_2^n}$ a dual solution.
	Then, for all $f : \F_2^n \to \F_2^n \text{ such that } f(\xv) = \Pm\xv + \vv$
	with $\Pm \in GL_n(\F_2)$ and $\vv\in\F_2^n$,
	$\bv_f = \left( b_{f(\iv)} \right)_{\iv\in\F_2^n}$ is also a dual solution.
\end{proposition}

The proof is deferred to Appendix \ref{Appendix:A}.

\begin{corollary}\label{Corollary:co-Hamming_bound}
	$$\sigma^L_{Av}(S) \le 2 \sum_{\iv \in \F_2^n} |\iv+\onev|_H|\halpha_\iv|^2.$$
\end{corollary}

\begin{proof}
	Let's show that $b_\iv =2|\iv+\onev|_H$ is a dual solution.
	By Theorem~\ref{Theorem:Hamming_bound}, $b_\iv=2|\iv|_H$ is a dual solution.
	Then, by Proposition~\ref{Proposition:dual_perm}, $b_\iv=2|\iv+\onev|_H$ is a dual solution (take $\Pm=I$ and $\vv=\onev$).
	The upper bound is derived immediately.
\end{proof}

\paragraph{The specific case of $n=2$}
For the specific case of $n=2$, we provide a complete characterization of the optimal dual solution. We show that it matches the bound of Theorem~\ref{Theorem:6} or of Proposition~\ref{Proposition:22}, up to symmetries on the indices. We prove this in Appendix~\ref{Appendix:n=2}. 

In the general case however, there are examples where none of these two bounds achieve the optimal solution.

\subsubsection{Matching primal solutions}
In section \ref{sec:avg_bounds}, we derive upper bounds on the primal program in the average setting using feasible solutions of the dual program.
Notice that our linear programs are parametrized by the Fourier coefficients of the states.
In this section, we give \emph{primal candidate solutions} associated to these dual solutions but they \emph{may not respect the nonnegativity constraints}.
For some parameters, that is, for some sets $S$ of symmetric states, these candidate primal solutions are nonnegative and by the complementary slackness theorem, this indicates that the corresponding upper bound is attained by the primal program.
Thus, we do not establish directly the nonnegativity of the candidate primal solutions, as our focus lies in the analysis of a parametric optimization problem.
We only say that the upper bound obtained from the dual equals the optimal value in the parameter region where the corresponding primal solution is nonnegative.

In this section only, we will use the following definitions that will be more convenient to state the matching primal solutions.
First, remark that a given information about $\xv$ is not uniquely determined by the choice of $(\Hm,\yv)$.
For instance, consider
$$ \Hm = \begin{pmatrix} 1 $ 0 $ 1 \\ 0 $ 1 $ 1 \end{pmatrix} \textrm{ and } \yv = \begin{pmatrix} 1 \\ 0 \end{pmatrix} \textrm{ then } \Hm \trsp{\xv} = \trsp{\yv} \Leftrightarrow x_1 \oplus x_3 = 1 \wedge x_2 \oplus x_3 = 0,$$
$$ \Hm' = \begin{pmatrix} 1 $ 0 $ 1 \\ 1 $ 1 $ 0 \end{pmatrix} \textrm{ and } \yv' = \begin{pmatrix} 1 \\ 1 \end{pmatrix} \textrm{ then } \Hm' \trsp{\xv} = \trsp{\yv'} \Leftrightarrow x_1 \oplus x_3 = 1 \wedge x_1 \oplus x_2 = 1.$$
One can easily check that $\Hm \trsp{\xv} = \trsp{\yv} \Leftrightarrow \Hm' \trsp{\xv} = \trsp{\yv'}$ which means the same information is described in two different ways. This comes from the fact that the lines of the two matrices $\Hm,\Hm'$ actually generate the same $2$-dimensional subspace of $\F_2^3$.
To circumvent this issue, we redefine our sets $\Lambda_k$ without duplicates.

\begin{definition}\label{Definition:LambdaTilde}
	For $0 \leq k \leq n$, we define $\Tilde{\Lambda}_k$ as the set of all parity-check matrices 
	$\Hm \in \F_2^{k \times n}$ of full rank $k$, taken \emph{without duplicates}, 
	meaning that two matrices $\Hm,\Hm'$ are identified whenever they generate the same code 
	$\{\xv \in \F_2^n \mid \Hm\xv = \zerov\}$.  
	Equivalently, $\Tilde{\Lambda}_k$ contains exactly one representative matrix for each 
	$[n,n-k]$ linear code over $\F_2$.
\end{definition}

We fix also a notation for the coset leader and we introduce a specific set of parity matrices.

\begin{definition}\label{def:coset_leader}
	We denote by $\rv_\Hm^{\text{min}}(\sv)$ and by $\rv_\Hm^{\text{MAX}}(\sv)$, the minimum weight representative of the dual coset $\D_\Hm(\sv)$ (called the coset leader) and the maximum weight representative of the dual coset $\D_\Hm(\sv)$ respectively.
	Formally,
	$$\rv_\Hm^{\text{min}}(\sv) = \arg\min_{\xv\in\D_\Hm(\sv)} |\xv|_H \quad \text{ and } \quad \rv_\Hm^{\text{MAX}}(\sv) = \arg\max_{\xv\in\D_\Hm(\sv)} |\xv|_H.$$
\end{definition}

\begin{definition}\label{Definition:mathcalE}
	Let \( n \in \mathbb{N} \). We define \( \mathcal{E}_k^n \) as the set of all ordered submatrices of the identity matrix \( I_n \), obtained by selecting \( k \) distinct rows of \( I_n \) in increasing order of their indices, for \( 0 \leq k \leq n \).
	Formally,
	$$
	\mathcal{E}_k^n = \left\{ 
	\begin{pmatrix}
		\mathbf{e}_{i_1} \\
		\mathbf{e}_{i_2} \\
		\vdots \\
		\mathbf{e}_{i_k}
	\end{pmatrix}
	\;\middle|\;
	1 \leq i_1 < i_2 < \cdots < i_k \leq n
	\right\},
	$$
	where \( \mathbf{e}_j \in \mathbb{F}_2^n \) denotes the \( j \)-th canonical basis vector.
\end{definition}

Furthermore, let $\onev$ = $1\dots1$ denote the all-ones vector. We are now in position to state our solutions.

\begin{proposition}[Hamming solution]\label{Proposition:primal_12}
	The dual solution $b_\iv = 2|\iv|_H \text{ for } \iv\in\F_2^n$
	is associated to the primal candidate solution
	\[ \lambda_\iv^\Hm = \left\{
	\begin{array}{ll}
		\frac{(-1)^{n-k} \sum_{\sv\in\F_2^{n-k}} (-1)^{|\sv|_H} |\halpha_{\rv_\Hm^{\text{MAX}}(\sv)}|^2}{|\halpha_\iv|^2}
		& \text{ for } k\in\Iint{0}{n}, \ \Hm\in\mathcal{E}_k^n, \ \iv\in \D_\Hm(\zerov), \\
		0 & \text{ otherwise}.
	\end{array}
	\right. \]
\end{proposition}

\begin{proposition}[co-Hamming solution]\label{Proposition:primal_1}
	The dual solution $b_\iv = 2|\iv+\onev|_H \text{ for } \iv\in\F_2^n$
	is associated to the primal candidate solution
	\[ \lambda_\iv^\Hm = \left\{
	\begin{array}{ll}
		\frac{(-1)^{n-k} \sum_{\sv\in\F_2^{n-k}} (-1)^{|\sv|_H} |\halpha_{\rv_\Hm^{\text{min}}(\sv)}|^2}{|\halpha_\iv|^2}
		& \text{ for } k\in\Iint{0}{n}, \ \Hm\in\mathcal{E}_k^n, \ \iv\in \D_\Hm(\onev), \\
		0 & \text{ otherwise}.
	\end{array}
	\right. \]
\end{proposition}

\begin{proposition}[Spike solution]\label{Proposition:primal_2}
	The dual solution $b_\zerov = 2^n+n-1 \text{ and } b_\iv = n-1 \text{ for } \iv\ne\zerov\in\F_2^n$
	is associated to the primal candidate solution
	$$ \lambda_\iv^\Hm = \left\{
	\begin{array}{ll}
		\frac{|\halpha_\zerov|^2}{|\halpha_\iv|^2} & \text{ for } \Hm=I, \ \iv\in\F_2^n, \\
		\frac{\sum_{\xv\in\D_\Hm(1)} |\halpha_{\xv}|^2 - \sum_{\xv\in\D_\Hm(0)} |\halpha_{\xv}|^2}{2^{n-1}|\halpha_{\iv}|^2} & \text{ for } \Hm\in\Tilde{\Lambda}_{n-1}, \ \iv\in\D_\Hm(\onev), \\
		0 & \text{ otherwise}.
	\end{array}
	\right. $$
\end{proposition}

The proofs are deferred to Appendix \ref{Appendix:A}.

\section{Efficient constructions}
Now, we investigate how the fine-grained unambiguous measurement can be implemented by quantum circuits given access to a primal solution ($\lambda_\iv^\Hm$). Our main observation is that the fine-grained unambiguous measurement reduces to the preparation of some controlled quantum state.
Assuming access to such a unitary, the fine-grained unambiguous measurement can be performed efficiently.

\begin{theorem}\label{thm:construction}
	Let $S = \{\ket{\psi_{\xv}}\}$ be a set of symmetric states with $\ket{\psi_{\zerov}} = \sum_{\iv \in \F_2^n} |\halpha_\iv|^2$ and let $(\lambda_\iv^{\Hm,\yv})$ be a primal solution, {\ie} an ensemble of nonnegative reals satisfying Equation~\ref{Eq:I12} and Equation~\ref{Eq:I1}. Assume that these numbers are efficiently quantum sampleable {\ie} that the unitary 
	
	\begin{equation}\label{eq:impl_U}
		U : \ket{\iv}\ket{0} \mapsto \sum_{k \in \Iint{0}{n}} \sum_{\Hm \in \Lambda_k} \sqrt{\lambda_\iv^\Hm} \frac{|\halpha_\iv|}{\halpha_\iv} \ket{\iv}\ket{\Hm},
	\end{equation}
	can be computed in time $poly(n)$. Then we can construct a POVM $\FHy \in \Gamma(S)$ such that 
	\begin{enumerate}
		\item $\rho(S;\FHy) = \rho^L(S;(\lambda_\iv^\Hm))$.
		\item The POVM $\FHy$ can be efficiently implemented in time $poly(n)$.
	\end{enumerate}
\end{theorem}

\begin{proof}
	Consider $k\in\Iint{0}{n}$, $\Hm \in\Lambda_k$, $\Gm_\Hm$ its associated matrix and $\sv\in\F_2^{n-k}$.
	By Theorem~$\ref{Theorem:LambdaRelations}$, for all $\iv\in\F_2^n$ such that $\iv\in\D_\Hm(\sv)$, there exist real numbers $\beta_{\sv}^{\Hm}$ such that the following relation is true:
	\begin{equation}
	\label{eq:impl_coset}
		\lambda_\iv^\Hm |\halpha_\iv|^2 = (\betav_\sv^\Hm)^2.
	\end{equation}
	Let $\sv\in\F_2^{n-k}$ and $\iv\in\F_2^n$ such that $\iv\in\D_\Hm(\sv)$, $\iv$ can be decomposed uniquely as $\iv = \Hm^\top \tv + \vv_\sv$ with $\tv\in\F_2^k$ and $\vv_\sv$ the coset leader of $\D_\Hm(\sv)$.
	Then for any $\xv \in \F_2^n$,
	\begin{equation}\label{eq:impl_trick}
		\xv\cdot\iv = \xv\cdot(\Hm^\top \tv+\vv_\sv) = (\Hm\xv)\cdot\tv+\xv\cdot\vv_\sv
	\end{equation}
	We define $\tilde{U} = (H^{\otimes n}\otimes I)U(H^{\otimes n}\otimes I)$.
	$$\tilde{U}: \ket{\hi}\ket{0} \mapsto \sum_{k \in \Iint{0}{n}} \sum_{\Hm \in \Lambda_k} \sqrt{\lambda_\iv^\Hm} \frac{|\halpha_\iv|}{\halpha_\iv} \ket{\hi}\ket{\Hm}.$$
	Recall that we can write each $\ket{\psi_\xv} = \sum_{\iv \in \F_2^n} (-1)^{\iv \cdot \xv} \halpha_\iv \ket{\hi}$. We therefore obtain
	\begin{align*}
		\tilde{U}\ket{\psi_\xv}\ket{0}
		= & \sum_{\iv\in\F_2^n} (-1)^{\iv\cdot\xv}\halpha_\iv \ket{\hi} \sum_{\Hm\in\Lambda} \sqrt{\lambda_\iv^\Hm} \frac{|\halpha_\iv|}{\halpha_\iv} \ket{\Hm} \\
		= & \sum_{\Hm\in\Lambda} \sum_{\sv\in\F_2^{n-k}} \betav_\sv^\Hm \left( \sum_{\iv\in\F_2^n \mid \iv\in\D_\Hm(\sv)} (-1)^{\iv \cdot \xv} \ket{\hx} \right) \ket{\Hm} & \text{from } \eqref{eq:impl_coset}
	\end{align*}
	We use the decomposition of $\iv$ and we apply the isometry $\ket{\hi} \mapsto \ket{\htv}\ket{\sv}$.
	\begin{align*}
		 & \sum_{\Hm\in\Lambda} \sum_{\sv\in\F_2^{n-k}} (-1)^{\xv\cdot\vv_\sv} \betav_\sv^\Hm \left( \sum_{\tv\in\F_2^k} (-1)^{\Hm\xv \cdot \tv} \ket{\htv} \right) \ket{\sv}\ket{\Hm} & \text{from } \eqref{eq:impl_trick} \\
		= & \sum_{\Hm\in\Lambda} \ket{\Hm\xv} \left(\sum_{\sv\in\F_2^{n-k}} (-1)^{\xv\cdot\vv_\sv} \sqrt{2^k} \betav_\sv^\Hm \ket{\sv} \right) \ket{\Hm} & \text{from } \eqref{def:action_qft}
	\end{align*}
	We can now measure all the qubits in the computational basis.
	The measurement outputs $\Hm, \Hm\xv$ with probability
	$$\Pr((\Hm,\Hm\xv)) = \sum_{\sv\in\F_2^{n-k}} 2^k(\betav_\sv^\Hm)^2 = \sum_{\sv\in\F_2^{n-k}} \sum_{\iv\in\D_\Hm(\sv)} (\betav_\sv^\Hm)^2 = \sum_{\sv\in\F_2^{n-k}} \sum_{\iv\in\D_\Hm(\sv)} \lambda_\iv^\Hm|\halpha_\iv|^2 = \sum_{\iv\in\F_2^n} \lambda_\iv^\Hm|\halpha_\iv|^2$$
	as the $\{\ket{\sv}\}_{\sv\in\F_2^{n-k}}$ are orthonormal.
\end{proof}

\section*{Acknowledgments} We acknowledge funding from the French PEPR integrated projects EPIQ (ANR-22-PETQ-007), PQ-TLS (ANR-22-PETQ-008) and HQI (ANR-22-PNCQ-0002) all part of plan France 2030 as well as the QuantERA QuantaGENOMICS project under Grant Agreement No. 101017733.

\newpage

\nocite{*}
\bibliographystyle{alpha}
\bibliography{biblio}

\appendix
\newpage

\section{Matching primal solutions for the average setting}\label{Appendix:A}

We first prove that from any dual solution $(b_\iv)_{\iv \in \F_2^n}$, one can also construct other dual solutions with the same objective by applying any affine bijection $f$ on the indices $\iv$.  
\begin{proposition}\label{Proposition:dual_perm}
	Let $\bv = (\bv_\iv)_{\iv \in \F_2^n}$ a dual solution.
	Then, for all $f : \F_2^n \to \F_2^n \text{ such that } f(\xv) = \Pm\xv + \vv$
	with $\Pm \in GL_n(\F_2)$ and $\vv\in\F_2^n$,
	$\bv_f = \left( \bv_{f(\iv)} \right)_{\iv\in\F_2^n}$ is also a dual solution.
\end{proposition}

\begin{proof}
	Let $\mathbf{b} = (b_\iv)_{\iv \in \F_2^n}$ be a dual solution, that is, for $\iv\in\F_2^n$, $b_\iv\in\mathbb{R}_{+}$ and for $k \in \Iint{0}{n}$, $\mathbf{H} \in \Lambda_k$, and $\mathbf{s} \in \F_2^{n-k}$, $\sum_{\iv \in D_\Hm(\sv)} b_\iv \geq C(k) 2^k$.
	We consider $f:\F_2^n \rightarrow \F_2^n$ such that $f(\iv) = \Pm\iv + \vv$ with $\Pm \in GL_n(\F_2)$ and $\vv\in\F_2^n$.
	Let us recall that $\D_\Hm(\sv) = \{\xv\in\F_2^n \mid \Gm_\Hm\xv=\sv\}$ and we define $f(\D_\Hm(\sv)) = \{f(\xv) \mid \xv\in\F_2^n, \ \Gm_\Hm \xv = \sv\} = \{\Pm\xv + \sv \mid \xv\in\F_2^n, \ \Gm_\Hm\xv=\sv\}$.
	We restate $f(\D_\Hm(\sv))$:
	\begin{align*}
		\yv\in f(\D_\Hm(\sv)) 
		& \Leftrightarrow \yv=\Pm\xv + \vv \text{ and } \Gm_\Hm\xv=\sv \\
		& \Leftrightarrow \xv=\Pm^{-1} (\yv + \vv) \text{ and } \Gm_\Hm \Pm^{-1} (\yv + \vv) = \sv \\
		& \Leftrightarrow \xv=\Pm^{-1} (\yv + \vv) \text{ and } \Gm_{\Hm'}\yv=\sv' \\
		& \Leftrightarrow \yv\in\D_{\Hm'}(\sv')
	\end{align*}
	with $\Hm' = \Hm\trp{\Pm}$, $\Gm_{\Hm'} = \Gm_\Hm\Pm^{-1}$ and $\sv'=\sv + \Gm_\Hm\Pm^{-1}\vv$.
	In other words, $f(\D_\Hm(\sv)) = \D_{\Hm'}(\sv')$.
	Finally, $$\sum_{\iv\in\D_\Hm(\sv)} \bv_{f(\iv)} = \sum_{\yv\in f(\D_\Hm(\sv))} \bv_\yv = \sum_{\yv\in\D_{\Hm'}(\sv')} \bv_{\yv} \geq C(k)2^k$$
\end{proof}

This proposition justifies that $b_\iv=2|\iv+\onev|$ is also a dual solution, since $b_\iv=2|\iv|$ is one.

\subsection{Hamming solution}

\begin{proposition}[Hamming solution]\label{Proposition:primal_12}
	The dual solution $b_\iv = 2|\iv|_H \text{ for } \iv\in\F_2^n$
	is associated to the primal candidate solution
	\[ \lambda_\iv^\Hm = \left\{
	\begin{array}{ll}
		\frac{(-1)^{n-k} \sum_{\sv\in\F_2^{n-k}} (-1)^{|\sv|_H} |\halpha_{\rv_\Hm^{\text{MAX}}(\sv)}|^2}{|\halpha_\iv|^2}
		& \text{ for } k\in\Iint{0}{n}, \ \Hm\in\mathcal{E}_k^n, \ \iv\in \D_\Hm(\zerov), \\
		0 & \text{ otherwise}.
	\end{array}
	\right. \]
\end{proposition}

\begin{proof}
	Same proof as for the co-Hamming solution below.
\end{proof}

\subsection{co-Hamming solution}
	We start the proof by giving some useful lemmas about the coset leader when $\Hm\in\mathcal{E}_k^n$. Recall this set was defined in Definition~\ref{Definition:mathcalE}.

	\begin{remark*}
		By definition of $\Hm\in\mathcal{E}_k^n$ means $\Hm$ is a subset of rows of the canonical basis.	
	\end{remark*}

	\begin{lemma}[Coset leader for $\Hm\in\mathcal{E}_k^n$]\label{lem:coset-leader-E}
		Let $\Hm\in\mathcal{E}_k^n$ and let $S\subseteq[n]=\{1,\dots,n\}$ be the index set of the rows of $\Hm$, so $|S|=k$ and $\Hm=I_S$.
		Let $C=\text{Ker}(\Hm)=\{\xv\in\F_2^n:\ \xv_{|S}=0\}$ and fix the canonical generator $\Gm_\C=I_{S^c}\in\F_2^{(n-k)\times n}$, where $S^c=[n]\setminus S$.
		For any $\sv\in\F_2^{\,n-k}$, the dual coset is
		$$\D_\Hm(\sv)=\{\,\xv\in\F_2^n:\ \Gm_\C\,\xv=\sv\,\}=\{\,\xv\in\F_2^n:\ \xv_{|S^c}=\sv\,\}.$$
		Its coset leader is $\yv \text{ with } \yv_{|S}=\zerov,\ \ \yv_{|S^c}=\sv.$
		In other words,
		$\yv=\sum_{j\in S^c:\ \sv_j=1}\mathbf \ev_j,$
		so the support of $\yv$ is contained in $S^c$.
	\end{lemma}

	\begin{proof}
		By construction, $\C=\text{Ker}(\Hm)=\{\xv:\ \xv_{|S}=0\}$ and $\Gm_\C=I_{S^c}$ maps $\xv$ to its restriction on $S^c$. Hence
		$\D_\Hm(\sv)=\{\xv:\ \xv_{|S^c}=\sv\}$.
		Any $\xv\in\D_\Hm(\sv)$ can be written uniquely as $\xv=(\uv,\sv)$ where $\uv\in\F_2^{\,k}$ fills the coordinates on $S$.
		The Hamming weight decomposes as $|\xv|_H=|\uv|_H+|\sv|_H$, which is minimized if and only if $\uv=0$.
		Therefore the unique minimum-weight element is $\rv_\Hm^{\text{min}}(\sv)$ defined by $\rv_\Hm^{\text{min}}(\sv)_{|S}=0$ and $\rv_\Hm^{\text{min}}(\sv)_{|S^c}=\sv$.
		Finally, since $\rv_\Hm^{\text{min}}(\sv)$ has ones exactly on the coordinates $j\in S^c$ with $\sv_j=1$, we have $\rv_\Hm^{\text{min}}(\sv)=\sum_{j\in S^c:\ \sv_j=1}\mathbf \ev_j$.
	\end{proof}

	\begin{lemma}\label{lem:unique_kH}
		Let $\xv\in\F_2^n$. There exists a unique pair
		$(k,\Hm)\in\{0,\dots,n\}\times\mathcal{E}_k^n$
		such that $\,\rv_\Hm^{\text{min}}(\onev)=\xv$ (and hence $\xv\in\D_\Hm(\onev)$). 
	\end{lemma}

	\begin{proof}
		Let $S=\{j\in[n]:\ x_j=0\}$ and set $k:=|S|=n-|\xv|_H$. Consider the matrix
		$\Hm  =I_S\in\mathcal{E}_k^n$, i.e., the $k\times n$ submatrix of the identity $I_n$ whose rows are $\{\mathbf e_j:\ j\in S\}$ in increasing order.
		For $\Hm\in\mathcal{E}_k^n$, we can use the previous lemma to compute the coset leader $\rv_\Hm^{\text{min}}(\sv)$.
		In particular, for $\sv=\onev\in\F_2^{n-k}$, we obtain $\rv_\Hm^{\text{min}}(\onev)$ is the vector equal to $0$ on $S$ and $1$ on $S^c$, which is exactly $\xv$. This proves existence, and also $\xv\in\D_\Hm(\onev)$ since $\xv_{|S^c}=\onev$.

		For uniqueness, suppose $\rv_{\widetilde\Hm}^{\text{min}}(\onev)=\xv$ with $\widetilde\Hm\in\mathcal{E}_{\tilde k}^n$ and let $\widetilde S$ be the index set of the rows of $\widetilde{\Hm}$. 
		By the explicit form above, $\rv_{\widetilde\Hm}^{\text{min}}(\onev)$ is $0$ on $\widetilde S$ and $1$ on $\widetilde S^c$. Since this vector equals $\xv$, we must have $\widetilde S=S$ and hence $\tilde k=|\widetilde S|=|S|=n-|\xv|_H=k$.
		Inside $\mathcal{E}_k^n$ the matrix with row set $S$ in increasing order is unique, namely $I_S=\Hm$. Therefore $(k,\Hm)$ is unique.
	\end{proof}

	\begin{lemma}\label{lem:N_k_i_x_correct}
	For $\iv,\xv\in\F_2^n$ and $k\in\{0,\dots,n\}$, define
	\[
	N(k,\iv,\xv)\ \eqdef\ 
	\bigl|\{\ \Hm\in\mathcal{E}_k^n\ :\ \iv\in\D_\Hm(\onev),\ \exists\,\sv\in\F_2^{n-k}\ \text{with}\ \rv_\Hm^{\text{min}}(\sv)=\xv\ \}\bigr|.
	\]
	Then
	\[
	N(k,\iv,\xv)
	=
	\begin{cases}
	\displaystyle \binom{\,|\iv|_H-|\xv|_H\,}{\,k-(n-|\iv|_H)\,} 
	& \text{if }\ \xv\subseteq\iv\ \text{and}\ n-|\iv|_H\le k\le n-|\xv|_H,\\[8pt]
	0 & \text{otherwise.}
	\end{cases}
	\]
	Here $\xv\subseteq\iv$ means $\mathrm{supp}(\xv)\subseteq\mathrm{supp}(\iv)$.
	\end{lemma}

	\begin{proof}
		Let $\iv,\xv\in\F_2^n$, let $k\in\{0,\dots,n\}$.
		Let $\Hm\in\mathcal{E}_k^n$ and write $S\subseteq[n]$ for the set of row indices of $\Hm$ and $S^c=[n]\setminus S$.
		With the canonical choice $\Gm_\C = I_{S^c}$ for $\C=\text{Ker}(\Hm)$, one has
		$$\D_\Hm(\sv)=\{\xv\in\F_2^n:\ x_{|S^c}=\sv\},\qquad \rv_\Hm^{\text{min}}(\sv)_{|S}=0,\ \rv_\Hm^{\text{min}}(\sv)_{|S^c}=\sv.$$
		We denote by $Z(\vv)\eqdef\{j\in[n]: v_j=0\}$ the zero-support of $\vv$.
		Hence, the two conditions in the definition of $N(k,\iv,\xv)$ translate to
		$$\iv\in\D_\Hm(\onev)\ \Longleftrightarrow\ \iv_{|S^c}=\onev\ \Longleftrightarrow\ S\supseteq Z(\iv),$$
		and
		$$\exists\,\sv\text{ with }\rv_\Hm^{\text{min}}(\sv)=\xv\ \Longleftrightarrow\ \xv_{|S}=0\ \Longleftrightarrow\ S\subseteq Z(\xv),$$
		Therefore, admissible $S$ are precisely those satisfying
		$$Z(\iv)\ \subseteq\ S\ \subseteq\ Z(\xv).$$
		In particular, this has a solution if and only if $Z(\iv)\subseteq Z(\xv)$, i.e., $\xv\subseteq\iv$.
		Assume $\xv\subseteq\iv$ and write
		$$Z(\xv)=Z(\iv)\ \cup\ \bigl(\mathrm{supp}(\iv)\setminus \mathrm{supp}(\xv)\bigr).$$
		Every admissible $S$ is then of the form $S=Z(\iv)\cup T$ with $T\subseteq \mathrm{supp}(\iv)\setminus \mathrm{supp}(\xv)$.
		The size constraint $|S|=k$ becomes $|T|=k-|Z(\iv)|=k-(n-|\iv|_H)$, which is feasible if and only if
		$$0\ \le\ k-(n-|\iv|_H)\ \le\ |\iv|_H-|\xv|_H \quad\Longleftrightarrow\quad n-|\iv|_H\ \le\ k\ \le\ n-|\xv|_H.$$
		In that case, the number of such $T$ is $\binom{|\iv|_H-|\xv|_H}{\,k-(n-|\iv|_H)\,}$, yielding exactly the claimed formula.
		If $\xv\not\subseteq\iv$, there is no admissible $S$ and the count is $0$.
	\end{proof}

	\begin{proposition}[co-Hamming solution]\label{Proposition:primal_1}
	The dual solution $b_\iv = 2|\iv+\onev|_H \text{ for } \iv\in\F_2^n$
	is associated to the primal candidate solution
	\[ \lambda_\iv^\Hm = \left\{
	\begin{array}{ll}
		\frac{(-1)^{n-k} \sum_{\sv\in\F_2^{n-k}} (-1)^{|\sv|_H} |\halpha_{\rv_\Hm^{\text{min}}(\sv)}|^2}{|\halpha_\iv|^2}
		& \text{ for } k\in\Iint{0}{n}, \ \Hm\in\mathcal{E}_k^n, \ \iv\in \D_\Hm(\onev), \\
		0 & \text{ otherwise}.
	\end{array}
	\right. \]
	\end{proposition}

	\begin{proof}
		Let $\iv=\zerov$,
		$\sum_{\Hm\in\Lambda} \lambda_\zerov^\Hm = \frac{|\halpha_\zerov|^2}{|\halpha_\zerov|^2} + \sum_{\Hm\in\Lambda^{*}} 0 = 1$ as for $\Hm\in\Lambda^{*}, \ \zerov\notin\D_\Hm(\onev)$.
		
		Let $\iv\ne\zerov\in\F_2^n$,
		\begin{align*}
			S \eqdef \sum_{\Hm\in\Lambda} \lambda_\iv^\Hm 
			& = \frac{1}{|\halpha_\iv|^2} \sum_{k=0}^{n} (-1)^{n-k} \sum_{\Hm\in\mathcal{E}_k^n \mid \iv\in \D_\Hm(\onev)} \sum_{\sv\in\F_2^{n-k}} (-1)^{|\sv|_H} |\halpha_{\rv_\Hm^{\text{min}}(\sv)}|^2
		\end{align*}
		Let's show that
		$$A \eqdef \sum_{k=0}^{n} (-1)^{n-k} \sum_{\Hm\in\mathcal{E}_k^n \mid \iv\in \D_\Hm(\onev)} \sum_{\sv\in\F_2^{n-k}} (-1)^{|\sv|_H} |\halpha_{\rv_\Hm^{\text{min}}(\sv)}|^2 = |\halpha_\iv|^2$$
		Let 
		$$ F_\iv(\xv) = \sum_{k=0}^n (-1)^{n-k} \sum_{\Hm\in\mathcal{E}_k^n \mid \iv\in \D_\Hm(\onev), \exists\sv\in\F_2^{n-k} \rv_\Hm^{\text{min}}(\sv) = \xv} (-1)^{|\sv|_H}$$
		and rewrite $A$ using $F_\iv(\xv)$, $$A = \sum_{x\in\F_2^n} |\halpha_\xv|^2 F_\iv(\xv).$$
		We now show that 
		\begin{align*}
			F_\iv(\xv) & = 1 \quad \text{if } \xv = \iv \\
			F_\iv(\xv) & = 0 \quad \text{otherwise}
		\end{align*}
		Note that, for $\Hm\in\mathcal{E}_k^n$, there is at most one $\sv\in\F_2^{n-k}$ such that $\rv_\Hm^{\text{min}}(\sv)=\xv$ since $\xv$ lies in one of the cosets of $\Hm$ but is not necessarily a coset leader.
		Assume $\xv=\iv$,
		\begin{align*}
			F_\iv(\iv)
			& = \sum_{k=0}^n (-1)^{n-k} \sum_{\Hm\in\mathcal{E}_k^n \mid \rv_{\Hm,\onev}=\iv} (-1)^{\onev\cdot\onev} \\
			& = |\{(k,\Hm)\in\Iint{0}{n}\times\mathcal{E}_k^n \mid \rv_{\Hm,\onev}=\iv \}| \\
			& = 1 & \text{from Lemma \ref{lem:unique_kH}} 
		\end{align*}
		Assume $\xv\ne\iv$,
		\begin{align*}
			F_\iv(\xv)
			& = \sum_{k=0}^n (-1)^{n-k} \sum_{\Hm\in\mathcal{E}_k^n \mid \iv\in \D_\Hm(\onev), \exists\sv\in\F_2^{n-k} \rv_\Hm^{\text{min}}(\sv) = \xv} (-1)^{|\sv|_H} \\
			& = (-1)^{|\sv|_H} \sum_{k=0}^{n} (-1)^{n-k} N(k,\iv,\xv) \\
			& = (-1)^{|\sv|_H} \sum_{k=n-|\iv|_H}^{n-|\xv|_H} (-1)^{n-k} \binom{|\iv|_H-|\xv|_H}{k-(n-|\iv|_H)} & \text{ from Lemma \ref{lem:N_k_i_x_correct}}\\
			& = (-1)^{|\sv|_H+|\iv|_H} \sum_{r=0}^{|\iv|_H-|\xv|_H} (-1)^{r} \binom{|\iv|_H-|\xv|_H}{r} \\
			& = 0
		\end{align*}
		This gives immediately that $A = |\halpha_\iv|^2$.
		Finally, if $\lambda_\iv^\Hm \ge 0$, complementary slackness conditions are satisfied as for all $\iv\in\F_2^n$, $(\sum_{\Hm} \lambda_\iv^\Hm - 1)b_\iv=0$.
		The primal and dual objective are equal, it is the optimum.
	\end{proof}

We can double check that the dual value is attained.
$$
	B \eqdef \sum_{k=0}^{n} \sum_{\Hm\in\Lambda_k} \sum_{\iv\in\F_2^n} k\lambda_\iv^\Hm|\halpha_\iv|^2 
	= \sum_{\iv\in\F_2^n} \sum_{k=0}^{n} k (-1)^{n-k} \sum_{\Hm\in\mathcal{E}_k^n \mid \iv\in \D_\Hm(\onev)} \sum_{\sv\in\F_2^{n-k}} (-1)^{|\sv|_H} |\halpha_{\rv_\Hm^{\text{min}}(\sv)}|^2
$$
We introduce $$F_\iv'(\xv) \triangleq \sum_{k=0}^n k (-1)^{n-k} \sum_{\Hm\in\mathcal{E}_k^n \mid \iv\in \D_\Hm(\onev), \exists\sv\in\F_2^{n-k} \rv_\Hm^{\text{min}}(\sv) = \xv} (-1)^{|\sv|_H}$$	
and we rewrite 
$$B = \sum_{\iv\in\F_2^n} \sum_{\xv\in\F_2^n} |\halpha_\xv|^2 F_\iv'(\xv)$$
Let's show that
\begin{equation*}
	F_\iv'(\xv) =
	\left\{
	\begin{array}{cl}
		|\iv+\onev|_H & \mbox{if } \xv=\iv \\
		1 & \mbox{if } \xv\subset\iv \text{ and } |\iv|_H-|\xv|_H=1 \\
		0 & \mbox{otherwise.}
	\end{array}
	\right.
\end{equation*}

\begin{align*}
	F_\iv'(\iv)
	& = \sum_{k=0}^n k(-1)^{n-k} \sum_{\Hm\in\mathcal{E}_k^n \mid \rv_{\Hm,\onev}=\iv} (-1)^{n-k} \\
	& = |\iv+\onev|_H & \text{ from Lemma \ref{lem:unique_kH}} \\
\end{align*}

$$F_\iv'(\xv) = \sum_{k=0}^n k(-1)^{n-k} \sum_{\Hm\in\mathcal{E}_k^n \mid \iv\in \D_\Hm(\onev), \exists\sv\in\F_2^{n-k} \rv_\Hm^{\text{min}}(\sv) = \xv} (-1)^{|\sv|_H}$$
Let $\Hm\in\mathcal{E}_k^n$ and write $S\subseteq[n]$ for the set of row indices of $\Hm$ and $S^c=[n]\setminus S$.
With the canonical choice $\Gm_\C = I_{S^c}$ for $\C=\text{Ker}(\Hm)$, one has
$$\D_\Hm(\sv)=\{\xv\in\F_2^n:\ x_{|S^c}=\sv\},\qquad \rv_\Hm^{\text{min}}(\sv)_{|S}=\xv_{|S}=0,\ \rv_\Hm^{\text{min}}(\sv)_{|S^c}=\xv_{|S^c}=\sv.$$
Thus, $|\xv|_H=|\sv|_H$.
\begin{align*}
	F_\iv'(\xv)
	& = (-1)^{|\xv|_H} \sum_{k=0}^{n} k (-1)^{n-k} N(k,\iv,\xv) \\
	& = (-1)^{|\xv|_H} \sum_{k=n-|\iv|_H}^{n-|\xv|_H} k (-1)^{n-k} \binom{|\iv|_H-|\xv|_H}{k-(n-|\iv|_H)} \\
	& = (-1)^{|\xv|_H+|\iv|_H} \sum_{r=0}^{|\iv|_H-|\xv|_H} (n-|\iv|_H+r) (-1)^{r} \binom{|\iv|_H-|\xv|_H}{r} \\
	& = (-1)^{|\xv|_H+|\iv|_H} \left[(n-|\iv|_H)\sum_{r=0}^{|\iv|_H-|\xv|_H} (-1)^{r} \binom{|\iv|_H-|\xv|_H}{r} + \sum_{r=0}^{|\iv|_H-|\xv|_H} r(-1)^{r} \binom{|\iv|_H-|\xv|_H}{r}\right] \\
	& = (-1)^{|\xv|_H+|\iv|_H} \sum_{r=0}^{|\iv|_H-|\xv|_H} r(-1)^{r} \binom{|\iv|_H-|\xv|_H}{r}
\end{align*}
Let $\xv\ne\iv$ such that $\xv\subset\iv$ and $|\iv|_H-|\xv|_H = 1$,
$F_\iv'(\xv) = (-1)(-1) = 1.$
Let $\xv\ne\iv$ such that $\xv\subset\iv$ and $|\iv|_H-|\xv|_H \ne \{0,1\}$,
$F_\iv'(\xv) = 0$.
Let $\xv\ne\iv$ such that $\xv\not\subset\iv$, $F_\iv'(\xv)=0$.
$$B = \sum_{\iv\in\F_2^n} |\iv + \onev||\halpha_\xv|^2 + \sum_{\iv\in\F_2^n} \sum_{\xv\in\F_2^n \mid \xv\subset\iv,\ |\iv|_H-|\xv|_H=1} |\halpha_\xv|^2$$
$$|\{ \xv\in\F_2^n \mid \xv\subset\iv,\ |\iv|_H-|\xv|_H=1 \}|=n-|\iv|_H=|\iv+\onev|_H$$
As expected, we obtain $$B = 2\sum_{\iv\in\F_2^n} |\iv +\onev ||\halpha_\iv|^2.$$

\subsection{Spike solution}

\begin{proposition}[Spike solution]\label{Proposition:primal_2}
	The dual solution $b_\zerov = 2^n+n-1 \text{ and } b_\iv = n-1 \text{ for } \iv\ne\zerov\in\F_2^n$
	is associated to the primal candidate solution
	\[ \lambda_\iv^\Hm = \left\{
	\begin{array}{ll}
		\frac{|\halpha_\zerov|^2}{|\halpha_\iv|^2} & \text{ for } \Hm=I, \ \iv\in\F_2^n, \\
		\frac{\sum_{\xv\in\D_\Hm(1)} |\halpha_{\xv}|^2 - \sum_{\xv\in\D_\Hm(0)} |\halpha_{\xv}|^2}{2^{n-1}|\halpha_{\iv}|^2} & \text{ for } \Hm\in\Tilde{\Lambda}_{n-1}, \ \iv\in\D_\Hm(1), \\
		0 & \text{ otherwise}.
	\end{array}
	\right. \]
\end{proposition}

\begin{proof}
	First, observe that for $\Hm = I, \ \F_2^n=\D_\Hm(0)$ and for $\Hm\in\Tilde{\Lambda}_{n-1}, \ \F_2^n=\D_\Hm(0) \sqcup \D_\Hm(1)$.

	Let $\iv=\zerov$,
	$\sum_{\Hm\in\Tilde{\Lambda}} \lambda_\zerov^\Hm = \frac{|\halpha_\zerov|^2}{|\halpha_\zerov|^2} + \sum_{\Hm\in\Tilde{\Lambda}_{n-1}} 0 = 1$ as for $\Hm\in\Tilde{\Lambda}_{n-1}, \ \zerov\notin\D_\Hm(1)$.
	
	Let $\iv\ne\zerov\in\F_2^n$,
	\begin{align*}
		S \eqdef \sum_{\Hm\in\Tilde{\Lambda}} \lambda_\iv^\Hm 
		& = \frac{|\halpha_\zerov|^2}{|\halpha_\iv|^2} + \frac{1}{2^{n-1}|\halpha_\iv|} \sum_{\Hm\in\Tilde{\Lambda}_{n-1} \mid \iv\in\D_\Hm(1)} \left( \sum_{\xv\in\D_\Hm(1)} |\halpha_\xv|^2 - \sum_{\xv\in\D_\Hm(0)} |\halpha_\xv|^2 \right) \\
		& = \frac{|\halpha_\zerov|^2}{|\halpha_\iv|^2} + \frac{1}{2^{n-1}|\halpha_\iv|} \sum_{\Hm\in\Tilde{\Lambda}_{n-1} \mid \iv\in\D_\Hm(1)} \left[ \left( |\halpha_\iv|^2 - |\halpha_\zerov|^2 \right) + \left( \sum_{\xv\in\D_\Hm(1)\setminus\{\iv\}} |\halpha_\xv|^2 - \sum_{\xv\in\D_\Hm(0)\setminus\{\zerov\}} |\halpha_\xv|^2 \right) \right] \\
	\end{align*}
	Let's show that 
	$$A \eqdef \sum_{\Hm\in\Tilde{\Lambda}_{n-1} \mid \iv\in\D_\Hm(1)} \left( \sum_{\xv\in\D_\Hm(1)\setminus\{\iv\}} |\halpha_\xv|^2 - \sum_{\xv\in\D_\Hm(0)\setminus\{\zerov\}} |\halpha_\xv|^2 \right) = 0.$$
	For $\Hm\in\Tilde{\Lambda}_{n-1}$, $\C=\text{Ker}(\Hm)$ of dimension $1$ by rank-nullity theorem.
	So, there exists $\cv_0\ne\zerov$ such that $\C=\text{span}\{\cv_0\}$.
	We choose $\Gm_\C = \trp{\cv}_0$.
	We rewrite $\D_\Hm(1) = \{\xv\in\F_2^n \mid \Gm_\C\cdot\xv = 1\} = \{\xv\in\F_2^n \mid \trp{\cv}_0\cdot\xv=1\}$ and $\D_\Hm(0) = \{\xv\in\F_2^n \mid \trp{\cv}_0\cdot\xv=0\}$.
	Remark that $\zerov\notin\D_\Hm(1)$ by definition and $\iv\not\in\D_\Hm(\zerov)$ as $\iv\in\D_\Hm(1)$.
	Then,
	\begin{align*}
		A & = \sum_{c_0\in\F_2^n \mid \trp{c}_0\cdot\iv = 1}\left(\sum_{\xv\in\F_2^n\setminus\{\iv\} \mid \trp{\cv}_0\cdot\xv=1}|\halpha_\xv|^2 - \sum_{\xv\in\F_2^n\setminus\{\zerov\} \mid \trp{\cv}_0\cdot\xv=0}|\halpha_\xv|^2 \right) \\
		  & = \sum_{\xv\in\F_2^n\setminus\{\zerov,\iv\}}|\halpha_\xv|^2 (N_1(\xv) - N_0(\xv))
	\end{align*}
	where we define for all $\xv\in\F_2^n\setminus\{\zerov,\iv\}$
	$$N_1(\xv) \eqdef |\{\cv_0\in\F_2^n \mid \trp{\cv}_0\cdot\iv = 1 \text{ and } \trp{\cv}_0\cdot\xv = 1\}|,$$
	$$N_0(\xv) \eqdef |\{\cv_0\in\F_2^n \mid \trp{\cv}_0\cdot\iv = 1 \text{ and } \trp{\cv}_0\cdot\xv = 0\}|.$$
	We are counting the number of vectors in the intersection of two affine hyperplanes.
	$\cv \mapsto \trp{\cv}\iv$ and $\cv \mapsto \trp{\cv}\xv$ are linearly independent as $\xv\ne\zerov,\iv$.
	Thus, $N_1(\xv) = N_0(\xv) = 2^{n-2}$ and $A=0$.
	\begin{align*}
		S & = \frac{|\halpha_\zerov|^2}{|\halpha_\iv|^2} + \frac{1}{2^{n-1}|\halpha_\iv|} \sum_{\Hm\in\Tilde{\Lambda}_{n-1}\mid\iv\in\D_\Hm(1)} (|\halpha_\iv|^2 - |\halpha_\zerov|^2) + 0 \\
		  & = \frac{|\halpha_\zerov|^2}{|\halpha_\iv|^2} + \frac{1}{2^{n-1}|\halpha_\iv|} (|\halpha_\iv|^2 - |\halpha_\zerov|^2)|\{\Hm\in\Tilde{\Lambda}_{n-1}\mid\iv\in\D_\Hm(1)\}|
	\end{align*}
	$\left\{\Hm\in\Tilde{\Lambda}_{n-1} \mid \iv\in\D_\Hm(1)\right\}$ defines an affine hyperplane of cardinality $2^{n-1}$.
	\begin{align*} 
		S & = \frac{|\halpha_\zerov|^2}{|\halpha_\iv|^2} + \frac{1}{2^{n-1}|\halpha_\iv|} 2^{n-1} (|\halpha_\iv|^2 - |\halpha_\zerov|^2) \\
		& = 1.
	\end{align*}
	Finally, if $\lambda_\iv^\Hm \ge 0$, complementary slackness conditions are satisfied as for all $\iv\in\F_2^n$, $(\sum_{\Hm} \lambda_\iv^\Hm - 1)b_\iv=0$.
	The primal and dual objective are equal, it is the optimum.
\end{proof}

We can double check that the dual value is attained.
\begin{align*}
	\sum_{k=0}^{n} \sum_{\Hm\in\Tilde{\Lambda}_k} \sum_{\iv\in\F_2^n} k\lambda_\iv^\Hm|\halpha_\iv|^2
	& = n2^n|\halpha_\zerov|^2 + \frac{n-1}{2^{n-1}} \sum_{\iv\in\F_2^n\setminus\{\zerov\}} \sum_{\Hm\in\Tilde{\Lambda}_{n-1} \mid \iv\in\D_\Hm(1)} \left( \sum_{\xv\in\D_\Hm(1)} |\halpha_{\xv}|^2 - \sum_{\xv\in\D_\Hm(0)} |\halpha_{\xv}|^2 \right) \\
	& = n2^n|\halpha_\zerov|^2 + (n-1)\sum_{\iv\in\F_2^n\setminus\{\zerov\}}(|\halpha_\iv|^2 - |\halpha_\zerov|^2) + 0 \\
	& = (n2^n - (n-1)(2^n-1))|\halpha_\zerov|^2 + (n-1)\sum_{\iv\in\F_2^n\setminus\{\zerov\}}|\halpha_\iv|^2 \\
	& = (2^n+n-1)|\halpha_\zerov|^2 + (n-1)\sum_{\iv\in\F_2^n\setminus\{\zerov\}} |\halpha_\iv|^2
\end{align*}

\section{Full characterization of fine-grained unambiguous measurements on $\F_2^2$ in the average setting}\label{Appendix:n=2}
In this section, we describe all the measurements that are possible on $2$ bits when considering the average number of parities setting.
The only parities we can learn on a 2 bits codeword $\xv=x_1x_2$ are $\xv=00,\ \xv=01,\ \xv=10,\ \xv=11,\ x_2=0,\ x_2=1,\ x_1=0,\ x_1=1,\ x_1\oplus x_2=0,\ x_1\oplus x_2=1,\ \emptyset$.
These parities are described using the set of matrices 
$$
\tilde{\Lambda}=\left\{
\begin{pmatrix}
    1 & 0 \\
    0 & 1 
\end{pmatrix},
\begin{pmatrix}
    0 & 1 
\end{pmatrix},
\begin{pmatrix}
    1 & 0
\end{pmatrix},
\begin{pmatrix}
    1 & 1
\end{pmatrix},
\begin{pmatrix}
    0 & 0
\end{pmatrix}
\right\}
$$
which correspond respectively to learning the full codeword, the first bit, the second bit, the xor or nothing. We use the notation of Definition~\ref{Definition:LambdaTilde}, where we consider in $\tilde{\Lambda}$ only one matrix $\Hm$ per affine subspace.

The corresponding unambiguous POVM is $\{F_{00}, F_{01}, F_{10}, F_{11}, F_{\bot0}, F_{\bot1}, F_{0\bot}, F_{1\bot}, F_{\text{xor}=0}, F_{\text{xor}=1}, F_\bot\}$.
When considering the states $\ket{\psi_\xv} = \halpha_{00}\ket{\widehat{00}} + (-1)^{c_2} \halpha_{01}\ket{\widehat{01}} + (-1)^{c_1} \halpha_{10}\ket{\widehat{10}} + (-1)^{c_1+c_2} \halpha_{11}\ket{\widehat{11}}$,
we assume without loss of generality that $|\halpha_{00}|^2 \leq |\halpha_{01}|^2 \leq |\halpha_{10}|^2 \leq |\halpha_{11}|^2$.

We first associate to the five matrices $\Hm$ in $\tilde{\Lambda}$ real numbers $\lambda_\iv^\Hm$ with $\iv \in \F_2^2$. We rewrite
\begin{align*}
	\lambda_{\iv}  = \lambda_{\iv}^{\begin{pmatrix}
			1 & 0 \\
			0 & 1 
			\end{pmatrix}} \ ; \ \mu_{\iv}  = \lambda_{\iv}^{\begin{pmatrix}
			0 & 1
		\end{pmatrix}} \ ; \ \nu_{\iv} = \lambda_{\iv}^{\begin{pmatrix}
		1 & 0 \end{pmatrix}} \ ; \ \xi_{\iv} = \lambda_{\iv}^{\begin{pmatrix}
			1 & 1 \end{pmatrix}} \ ; \ \delta_{\iv} = \lambda_{\iv}^{\begin{pmatrix}
				0 & 0 \end{pmatrix}}.
\end{align*} 

We now rewrite the linear program associated to the objective $\rho^L_{Av}(S)$. We first write the linear relations between the variables. In our case, we obtain

    $$\lambda_{00} |\halpha_{00}|^2 = \lambda_{01} |\halpha_{01}|^2 = \lambda_{10} |\halpha_{10}|^2 = \lambda_{11} |\halpha_{11}|^2$$
   as well as 
\begin{align*}
	\mu_{00} |\halpha_{00}|^2 = \mu_{10} |\halpha_{10}|^2
	&& \mu_{01} |\halpha_{01}|^2 = \mu_{11} |\halpha_{11}|^2 \\
	\nu_{00} |\halpha_{00}|^2 = \nu_{01} |\halpha_{01}|^2
	&& \nu_{10} |\halpha_{10}|^2 = \nu_{11} |\halpha_{11}|^2 \\
	\xi_{00} |\halpha_{00}|^2 = \xi_{11} |\halpha_{11}|^2
	&& \xi_{01} |\halpha_{01}|^2 = \xi_{10} |\halpha_{10}|^2
\end{align*}
We plug these relations in the linear program in order to reduce the number of variables. We thus obtain the following linear program


\begin{mdframed}[linewidth=1pt, roundcorner=5pt, nobreak=true, innerleftmargin=10pt, innerrightmargin=10pt, innertopmargin=8pt, innerbottommargin=8pt]
	\begin{center}
		\textbf{$2$-bits primal program}
	\end{center}
	\vspace{0.5em}
    Variables:
    $$\lambda_{00}, \mu_{00}, \mu_{01}, \nu_{00}, \nu_{10}, \xi_{00}, \xi_{01}, \delta_{00}, \delta_{01}, \delta_{10}, \delta_{11} \in \mathbb{R}_{+}$$
    Objectives:
    $$\rho^L(S) \eqdef \max 8\lambda_{00}|\halpha_{00}|^2 + 2(\mu_{00}|\halpha_{00}|^2+\mu_{10}|\halpha_{01}|^2) + 2(\nu_{00}|\halpha_{00}|^2+\nu_{10}|\halpha_{10}|^2) + 2(\xi_{00}|\halpha_{00}|^2+\xi_{01}|\halpha_{01}|^2)$$
    Constraints:
    \begin{align*}
        1 - \lambda_{00} - \mu_{00} - \nu_{00} - \xi_{00} - \delta_{00} &=0\\
        1 - \frac{|\halpha_{00}|^2}{|\halpha_{01}|^2}\lambda_{00} - \mu_{01} - \frac{|\halpha_{00}|^2}{|\halpha_{01}|^2}\nu_{00} - \xi_{01} - \delta_{01} &=0\\
        1 -\frac{|\halpha_{00}|^2}{|\halpha_{10}|^2}\lambda_{00} - \frac{|\halpha_{00}|^2}{|\halpha_{10}|^2}\mu_{00} - \nu_{10} - \frac{|\halpha_{01}|^2}{|\halpha_{10}|^2}\xi_{01} - \delta_{10} &=0\\
        1 -\frac{|\halpha_{00}|^2}{|\halpha_{11}|^2}\lambda_{00} - \frac{|\halpha_{01}|^2}{|\halpha_{11}|^2}\mu_{01} - \frac{|\halpha_{00}|^2}{|\halpha_{10}|^2}\nu_{10} - \frac{|\halpha_{00}|^2}{|\halpha_{11}|^2}\xi_{00} - \delta_{11} &=0\\
	\end{align*}
\COMMENT{	\begin{alignat*}{6}
		1 &{}-{}& \lambda_{00} 
		&{}-{}& \mu_{00} 
		&{}-{}& \nu_{00} 
		&{}-{}& \xi_{00} 
		&{}-{}& \delta_{00} &= 0\\[4pt]
		1 &{}-{}& \frac{|\halpha_{00}|^2}{|\halpha_{01}|^2}\lambda_{00} 
		&{}-{}& \mu_{01} 
		&{}-{}& \frac{|\halpha_{00}|^2}{|\halpha_{01}|^2}\nu_{00} 
		&{}-{}& \xi_{01} 
		&{}-{}& \delta_{01} &= 0\\[4pt]
		1 &{}-{}& \frac{|\halpha_{00}|^2}{|\halpha_{10}|^2}\lambda_{00} 
		&{}-{}& \frac{|\halpha_{00}|^2}{|\halpha_{10}|^2}\mu_{00} 
		&{}-{}& \nu_{10} 
		&{}-{}& \frac{|\halpha_{01}|^2}{|\halpha_{10}|^2}\xi_{01} 
		&{}-{}& \delta_{10} &= 0\\[4pt]
		1 &{}-{}& \frac{|\halpha_{00}|^2}{|\halpha_{11}|^2}\lambda_{00} 
		&{}-{}& \frac{|\halpha_{01}|^2}{|\halpha_{11}|^2}\mu_{01} 
		&{}-{}& \frac{|\halpha_{00}|^2}{|\halpha_{10}|^2}\nu_{10} 
		&{}-{}& \frac{|\halpha_{00}|^2}{|\halpha_{11}|^2}\xi_{00} 
		&{}-{}& \delta_{11} &= 0
	\end{alignat*}}
    \COMMENT{$$\lambda_{00} |\halpha_{00}|^2 = \lambda_{01} |\halpha_{01}|^2 = \lambda_{10} |\halpha_{10}|^2 = \lambda_{11} |\halpha_{11}|^2$$
    \begin{align*}
        \mu_{00} |\halpha_{00}|^2 = \mu_{10} |\halpha_{10}|^2
        && \mu_{01} |\halpha_{01}|^2 = \mu_{11} |\halpha_{11}|^2 \\
        \nu_{00} |\halpha_{00}|^2 = \nu_{01} |\halpha_{01}|^2
        && \nu_{10} |\halpha_{10}|^2 = \nu_{11} |\halpha_{11}|^2 \\
        \xi_{00} |\halpha_{00}|^2 = \xi_{11} |\halpha_{11}|^2
        && \xi_{01} |\halpha_{01}|^2 = \xi_{10} |\halpha_{10}|^2
    \end{align*}}
\end{mdframed}

From this program, we construct the associated dual linear program. 
\begin{mdframed}[linewidth=1pt, roundcorner=5pt, nobreak=true, innerleftmargin=10pt, innerrightmargin=10pt, innertopmargin=8pt, innerbottommargin=8pt]
	\begin{center}
		\textbf{2-bits dual program}
	\end{center}
	\vspace{0.5em}
	Variables:
    $$b_{00}, b_{01}, b_{10}, b_{11} \in \mathbb{R}_+$$
	Objective:
    $$\sigma^L(S) \eqdef \min b_{00}|\halpha_{00}|^2 + b_{01}|\halpha_{01}|^2 + b_{10}|\halpha_{10}|^2 + b_{11}|\halpha_{11}|^2$$
	Constraints:
    $$b_{00} + b_{01} + b_{10} + b_{11} \geq 8$$
	\begin{align*}
        b_{00} + b_{10} \geq 2 
        && b_{01} + b_{11} \geq 2 \\
        b_{00} + b_{01} \geq 2
        && b_{10} + b_{11} \geq 2 \\
        b_{00} + b_{11} \geq 2
        && b_{01} + b_{10} \geq 2
	\end{align*}
\end{mdframed}


By Corollary~\ref{Corollary:co-Hamming_bound} and Proposition~\ref{Proposition:primal_1}, we obtain that

\begin{corollary}[co-Hamming solution]
    The optimal valus is $4|\halpha_{00}|^2 + 2|\halpha_{01}|^2 + 2|\halpha_{10}|^2$ when
    \begin{align*}
    \lambda_{00} = 1, \\
    \mu_{00} = 0, & \quad \mu_{01} = \frac{|\halpha_{01}|^2-|\halpha_{00}|^2}{|\halpha_{01}|^2},\\
    \nu_{00} = 0, & \quad \nu_{10} = \frac{|\halpha_{10}|^2-|\halpha_{00}|^2}{|\halpha_{10}|^2},\\
    \xi_{00} = 0, & \quad \xi_{01} = 0
    \end{align*}
    and 
    $$\delta_{00} = 0, \quad \delta_{01} = 0, \quad \delta_{10} = 0, \quad \delta_{11} = \frac{|\halpha_{00}|^2+|\halpha_{11}|^2-|\halpha_{01}|^2-|\halpha_{10}|^2}{2|\halpha_{11}|^2}$$
    are nonnegative.
\end{corollary}

By Proposition~\ref{Proposition:22} and Proposition~\ref{Proposition:primal_2}, we obtain that

\begin{corollary}[Spike solution]
    The optimal valus is $5|\halpha_{00}|^2 + |\halpha_{01}|^2 + |\halpha_{10}|^2 + |\halpha_{11}|^2$ when
    \begin{align*}
    \lambda_{00} = 1, \\
    \mu_{00} = 0, & \quad \mu_{01} = \frac{|\halpha_{01}|^2+|\halpha_{11}|^2-|\halpha_{00}|^2-|\halpha_{10}|^2}{2|\halpha_{01}|^2},\\
    \nu_{00} = 0, & \quad \nu_{10} = \frac{|\halpha_{10}|^2+|\halpha_{11}|^2-|\halpha_{00}|^2-|\halpha_{01}|^2}{2|\halpha_{10}|^2},\\
    \xi_{00} = 0, & \quad \xi_{01} = \frac{|\halpha_{01}|^2+|\halpha_{10}|^2-|\halpha_{00}|^2-|\halpha_{11}|^2}{2|\halpha_{01}|^2},
    \end{align*}
    and
    $$\delta_{00} = 0, \quad \delta_{01} = 0, \quad \delta_{10} = 0, \quad \delta_{11} = 0$$
    are nonnegative.
\end{corollary}

\begin{corollary}
    Up to parameters permutations, these are the only two families of 2-bits measurement possible in the average setting.
\end{corollary}

\begin{proof}
    We assume without loss of generality that $|\halpha_{00}|^2 \leq |\halpha_{01}|^2 \leq |\halpha_{10}|^2 \leq |\halpha_{11}|^2$.
    By assumptions, we always have that $|\halpha_{01}|^2 - |\halpha_{00}|^2 \geq 0$ and $|\halpha_{10}|^2 - |\halpha_{00}|^2 \geq 0$.
    Moreover, $|\halpha_{10}|^2 + |\halpha_{11}|^2 \geq |\halpha_{00}|^2 + |\halpha_{11}|^2$ and $|\halpha_{01}|^2 + |\halpha_{11}|^2 \geq |\halpha_{00}|^2 + |\halpha_{11}|^2 \geq |\halpha_{00}|^2 + |\halpha_{10}|^2$.
    Finally, either $|\halpha_{00}|^2 + |\halpha_{11}|^2 \geq |\halpha_{01}|^2 + |\halpha_{10}|^2$ and the co-Hamming solution is nonnegative and the optimal value is $4|\halpha_{00}|^2 + 2|\halpha_{01}|^2 + 2|\halpha_{10}|^2$
    or, $|\halpha_{00}|^2 + |\halpha_{11}|^2 \leq |\halpha_{01}|^2 + |\halpha_{10}|^2$ and the spike solution is nonnegative and the optimal value is $5|\halpha_{00}|^2 + |\halpha_{01}|^2 + |\halpha_{10}|^2 + |\halpha_{11}|^2$.
\end{proof}

We can then generalize this analysis in the case we don't have $|\halpha_{00}|^2 \leq |\halpha_{01}|^2 \leq |\halpha_{10}|^2 \leq |\halpha_{11}|^2$ by sorting the indices with respect to the largest $|\halpha_{\iv}|^2$ and then performing the same analysis. 

\COMMENT{
\begin{figure}[H]
	\centering
	\begin{minipage}{0.4\textwidth}
		\centering
		\includegraphics[width=\textwidth]{polytope.png}
		\caption{Projection of the polytope in the subspace $(1, \mu_{01}, \nu_{10}, \delta_3)$ in the case $|\halpha_0|^2 \leq |\halpha_1|^2 \leq |\halpha_2|^2 \leq |\halpha_3|^2$ and $|\halpha_1|^2 + |\halpha_2|^2 \leq |\halpha_0|^2 + |\halpha_3|^2$. Plotted with $|\halpha|^2=[0.05, 0.15, 0.3, 0.5]$.}
		\label{fig:polytope}
	\end{minipage}
	\hfil
	\begin{minipage}{0.4\textwidth}
		\centering
		\includegraphics[width=\textwidth]{polytope2.png}
		\caption{Projection of the polytope in the subspace $(1, \mu_{01}, \nu_{10}, \xi_{01})$ in the case $|\halpha_0|^2 \leq |\halpha_1|^2 \leq |\halpha_2|^2 \leq |\halpha_3|^2$ and $|\halpha_0|^2 + |\halpha_3|^2 \leq |\halpha_1|^2 + |\halpha_2|^2$. Plotted with $|\halpha|^2=[0.05, 0.3, 0.3, 0.35]$.}
		\label{fig:polytope2}
	\end{minipage}
\end{figure}
}

\end{document}